\documentclass[]{interact}

\usepackage{epstopdf}
\usepackage[caption=false]{subfig}

\usepackage[svgnames]{xcolor}
\usepackage{tabularx}
\usepackage{chngpage}
\usepackage[ruled, linesnumbered, commentsnumbered]{algorithm2e}

\usepackage[longnamesfirst,sort]{natbib}
\bibpunct[, ]{(}{)}{;}{a}{,}{,}
\renewcommand\bibfont{\fontsize{10}{12}\selectfont}

\theoremstyle{plain}

\newtheorem{lemma}{Lemma}[subsection]

\theoremstyle{definition}

\theoremstyle{remark}

\begin{document}

\articletype{RESEARCH ARTICLE} 

\title{A Game Theoretic Framework for Surplus Food Distribution in Smart Cities and Beyond}

\author{
\name{Surja Sanyal\textsuperscript{a}\thanks{CONTACT Surja Sanyal. Email: hi.surja06@gmail.com}, Vikash Kumar Singh\textsuperscript{b}\thanks{CONTACT Vikash Kumar Singh. Email: vikash.singh@vitap.ac.in}, Fatos Xhafa\textsuperscript{c}\thanks{CONTACT Fatos Xhafa. Email: fatos@cs.upc.edu}, Banhi Sanyal\textsuperscript{d}\thanks{CONTACT Banhi Sanyal. Email: banhisanyal9@gmail.com}, and Sajal Mukhopadhyay\textsuperscript{e}\thanks{CONTACT Sajal Mukhopadhyay. Email: sajal@cse.nitdgp.ac.in}}
\affil{\textsuperscript{a}National Institute of Technology, Durgapur, West Bengal, India; \textsuperscript{b}School of Computer Science and Engineering, VIT-AP University, Amaravati, Andhra Pradesh, India; \textsuperscript{c}Universitat Polit\`{e}cnica de Catalunya, Barcelona, Spain; \textsuperscript{d}National Institute of Technology, Rourkela, Odisha, India; \textsuperscript{e}National Institute of Technology, Durgapur, West Bengal, India}
}

\maketitle

\begin{abstract}
Food waste is a major challenge for the present world. It is the precursor to several socioeconomic problems that are plaguing the modern society. To counter the same and to, simultaneously, stand by the undernourished, surplus food redistribution has surfaced as a viable solution. Information and Communications Technology (ICT)-mediated food redistribution is a highly scalable approach and it percolates into the masses far better. Even if ICT is not brought into the picture, the presence of food surplus redistribution in developing countries like India is scarce and is limited to only a few of the major cities. The discussion of a surplus food redistribution framework under strategic settings is a less discussed topic around the globe. This paper aims at addressing a surplus food redistribution framework under strategic settings, thereby facilitating a smoother exchange of surplus food in the smart cities of developing countries, and beyond. As ICT is seamlessly available in smart cities, the paper aims to focus the framework in these cities. However, this can be extended beyond the smart cities to places with greater human involvement.
\end{abstract}

\begin{abbreviations}

\noindent CA-DTB: Chronological Acceptance using Double Tie Breaking;\\
CTFU: Classification and TriFurcation of Users;\\
FDRM-CA: Food Donor to Receiver Matching with Chronological Acceptance;\\
FSSAI: Food Safety and Standards Authority of India;\\
GPS: Global Positioning System;\\
ICT:  Information and Communications Technology;\\
IFSA: Indian Food Sharing Alliance;\\
NGO: Non Government Organizations;\\
NUIR: New User Interrupt Routine;\\
P2P: Peer-To-Peer;\\
UT: Union Territory

\end{abbreviations}


\section{Introduction}

For the last couple of years, food waste has been a major contributor towards several socioeconomic problems \citep{ganglbauer2012creating, harvey2020food}, including, but not limited to, global warming, greenhouse gas emissions, water wastage, soil degradation, farmer suicides, price fluctuations, black marketing, and hoarding. On a global scale, more than 30\% of all edible products ends up as waste \citep{shi2020improving}. Each year this wastage amounts to billions of tons of food, out of which near 60\% is avoidable \citep{zhong2017food}. One of the solutions that can deal with this underrated, yet malicious, problem of food wastage is surplus food redistribution. Distributing food from the surplus to the deficit hits the generation of food waste at the surplus end and hunger at the deficit end. While this may not be a solution to the whole situation, it targets the 60\% avoidable food waste into a win-win scenario for the society as well as the economy of the land. Information and Communications Technology (ICT) has taken this a step further \citep{ciaghi2016beyond} and made the redistribution activity smoother by connecting the donors with the receivers, and the available volunteers, by minimizing human interference as far as possible. The use of ICT brings the very desirable properties of sustainability and, more importantly, scalability to the table. ICT-mediation in food redistribution is the buzz of the day \citep{weymes2018disruptive} and is gaining fast popularity around the world due to its numerous upsides. In the smart cities of developing countries like India, however, a lot remains to be done as the presence of food redistribution movements mediated by ICT is scarce and only limited to a few major cities. The discussion of a surplus food redistribution framework under strategic settings is a less discussed topic around the globe. This paper aims at addressing a surplus food redistribution framework under strategic settings, thereby facilitating a smoother exchange of surplus food in the smart cities of developing countries, and beyond. As ICT is seamlessly available in smart cities, the paper aims to focus the framework in these cities. However, this can be extended beyond the smart cities to places with greater human involvement.

\hyperref[Fig-5]{\textbf{Figure 1}} presents an overview of the model. The donation that the donors make is transported by the volunteers, who are available on the fly or as planned, to the receivers for the matching assignment made by the algorithm, as shown by the \textcolor{Brown}{\textbf{brown arrows}}. The movement of information used by the app for match generation has been denoted by \textcolor{DarkBlue}{\textbf{blue arrows}}, and the generated match information flow has been represented with \textcolor{DarkGreen}{\textbf{green arrows}}. Notice the \textcolor{DarkMagenta}{\textbf{magenta arrows}} that indicate that donors and receivers can also volunteer for transportation, while being a donor or a receiver, simultaneously.

The rest of the paper is organized into the following sections. Section \ref{Sec-2} presents the related work in this field; Section \ref{Sec-3} presents the method followed in presenting the solution to our problem. The system model has been discussed in Section \ref{Sec-4}. The mechanism used is discussed in Section \ref{Sec-5} followed by its analysis and simulation in Section \ref{Analysis}. Conclusion and future work is presented in Section \ref{Sec-7}.

\begin{figure}
    \centering
    \includegraphics[width=\textwidth]{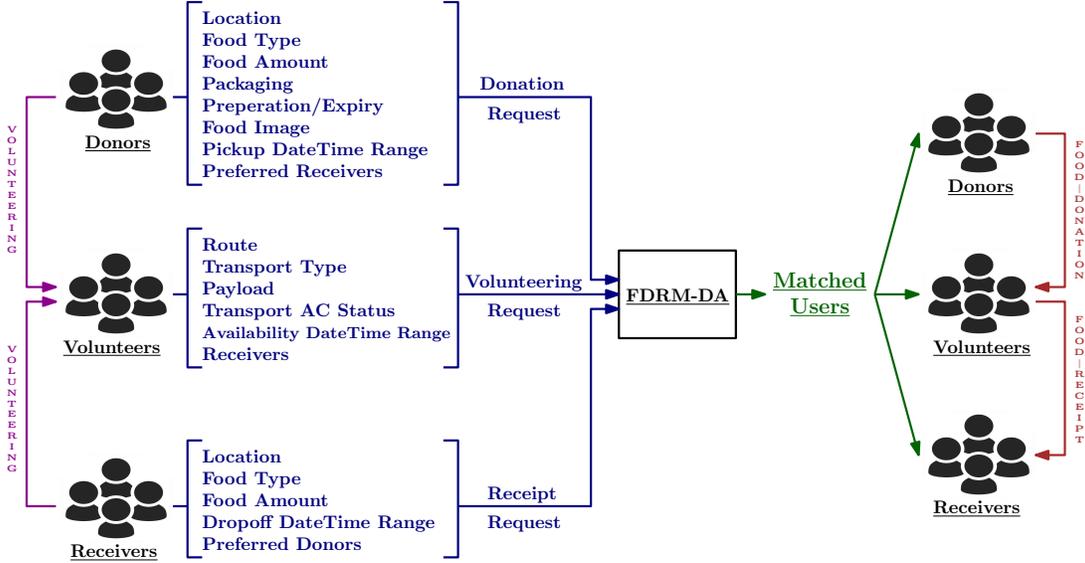}
    \caption{Overview of the Model} \label{Fig-5}
\end{figure}

\section{Related Work} \label{Sec-2}

The foundation of our work already exists in the form of establishing food redistribution as one of the leading solutions of dealing food wastage, and ICT-mediation as a great add-on to boost the surplus movement. However, work on the algorithms used to match the donors with the receivers remains elusive and is the main agenda of our paper.

\subsection{Literature Review}

A study \citep{spring2020capturing} has identified the rapid global growth of surplus food redistribution initiatives relying on volunteer activities mediated by ICT-based platforms. The involvements and contributions of cultures, organisational facts and internal experiences has also been explored by another study \citep{rut2020participating}. Reduction of food waste production as an approach towards mitigating climate change \citep{nikravech2020limiting} attempted by grassroots initiatives aim to prevent avoidable food waste and redistribute surplus food for consumption. These low budget initiatives rely heavily on volunteers to curb their costs. Another study \citep{bergstrom2020sustainability} suggests that environmental sustainability is best supported by an approach called the \textit{Social Supermarket} in which a body for handling the surplus of a city works with several food distributors gets its surplus either directly from the donors or through a government charitable body and redistributes it to the needy. A unique study \citep{koivunen2020increasing} reports an attempt at the tricky task of reducing consumer food waste at cafeterias using the combo of ICTs and IoTs. Lunch lines containing sensors that track food wastage which is later showed in each consumer's mobile via an app revealed an average 3\% food wasted by the customers. Another study \citep{spicer2020surplus} worked to create surplus food data of market value and assess the prevailing practices of surplus food management. As mentioned, partnering with a food sharing app \textit{OLIO} (\url{https://olioex.com/}), one study \citep{harvey2020food} presents an analysis of the social network of the app. It reveals the formation of new relationships that divert from the expected donor-receiver ones. This may prove pivotal for policymakers who aim to analyse and encourage ICT-assisted food redistribution. Recent studies \citep{eriksson2020there, davies2019food} discuss the legislation policies affecting food wastage directly or indirectly in the European Union. They point out that the food policies of the land affect proper food redistribution and food waste prevention.

A study \citep{weymes2019re} points out that the use of ICT speeds-up the redistribution process, provides high scalability of the operations and also attracts better quality surplus food. Research \citep{berns2019commodities} also documents the complexities and opportunities involved in food redistribution locally. It also highlights the practices that can transform the way surplus food is perceived and revalued. There also has been research \citep{durr2019exploring} on a specific South African mobile app named \textit{Food for Us} (\url{http://foodforus.co.za/}). This app works on providing food to the under-provided and allowing small-scale producers to get their produce to markets. This paper reveals the need for developing a strong social networked system build around technological platforms such as this app, to help find alternate markets for unsold farm produce. Another study \citep{spring2019sites} provides valuable learning by contrasting images between different food handling practices. It ponders on the need of knowing and paying attention to the details that that help differentiate between food and waste. One study \citep{arcuri2019food} portrays food redistribution as a double-sided sword to fight against food insecurity and food waste. It discusses the efficacy of the anti-waste/pro-donation law of Italy in addressing both food waste and insecurity. It points out that its effectiveness lies in bringing together the different actors in the food management process for tackling the food insecurity issue. Another study \citep{king2019identifying} focuses on the concept of a peer-to-peer (P2P) food sharing system based on high internet usage and people interested in food sharing. This study acts as a base to identify future P2P food collection and redistribution zones and the actors they can attract. A different study \citep{light2019platforms} analyzes the creation and accumulation of networks of sharing individuals, and how digital platforms can aid in the scaling of such an approach keeping in mind the sustainability of such initiatives. One study \citep{davies2019urban} has directly assessed the scenario of sharing food in cities and has viewed it as a field of experimentation and innovation.

It has been recognised \citep{el2018transition} that ICTs do contribute towards a smoother transformation to sustainable food systems by increasing resource productivity, reducing inefficiencies, decreasing management costs, and improving food chain coordination. As a matter of fact, a study \citep{davies2018re} observes the contribution of ICT in the redistribution of surplus prepared or cooked food which has a very short window of recoverability before it spoils. One study \citep{rut2018sharing} terms food sharing practices as ``messy" pointing out that it includes diverse ranges of participants and practices that also vary over space and time and get connected via both physical and virtual platforms (ICTs). As noted previously, a study \citep{weymes2018disruptive} points out that ICT intervention resulted in the recovery of surplus food at all stages of food production. However, maximum recovery was noted in the retail stage. Research \citep{facchini2018food} finds production beyond need, poor management, and bad consumer behaviours as the reasons that give birth to the difficult challenge that is food wastage on a global scale. One study \citep{frigo2018working} explores how food donation facilitates the transition towards a circular economy and bring together the diverse players of the process. It draws an important conclusion that multi-agent collaborations are the key towards a circular economy. A thesis \citep{rombach2018new} presents the different motives and interactions of the different actors involved in the act of food distribution. Another study \citep{davies2018fare} summarizes the food redistribution initiatives taken by individuals across more than ninety countries. It evaluates these initiatives in terms of sustainability and creates a database that depicts the transformation of the society towards sustainability.

Research \citep{davies2017making} acknowledges the positive impact of ICT in food redistribution and also creates a database of such food sharing activities across several cities, countries and continents in order to facilitate the identification and analysis of repetitive patterns and temporal trends in ICT-mediated surplus redistribution. Food sharing, coupled with structural changes along the production line and better consumer habits, has been deemed \citep{falcone2017bringing} as a potential solution towards reducing food waste. A study \citep{mousa2017organizations} establishes the fact that more than one out of eight people in the United States is affected with food insecurity in spite of around one thirds of all food production going into landfills. It acknowledges that food redistribution organizations do mitigate this insecurity, although more in some states of the country than in others. Another study \citep{vittuari2017second} recognizes the positive impacts of food redistribution on the environment, economy, and the society in Italy. A different study \citep{berry2017sharing} acknowledges the contrasting challenges of food waste and food insecurity in the US. Based on the analysis of K–12 school system in Maine, it points out that very little is known about the strategies that schools use to handle the considerable amount of food waste that they produce. A study \citep{ananprakrit2017traceability} studied the prevailing traceability practices of Stockholm’s Stadsmission food bank. Although the current practices proved to be sufficient to provide food safety and quality, potential for systematic errors was also detected. A study \citep{philip2017technical} observes the working of an Israeli food bank using a different logistical model from the other food banks. They have nonprofit organizations (NPOs) as intermediaries that add fresh produce to the surplus thereby improving the food value of the redistribution. One study \citep{fleming2017estimating} contrasts the environmental impact of food redistribution from donation to consumption against that of alternative options like landfilling and composting. Although the facts favoured composting, a better optimized food redistribution, utilizing benefits of using the other alternative options in the process, seems like the win-win solution of the future.

A study \citep{jurgilevich2016transition} accepts that the unwise and inefficient use of food resources has rendered us in need for a transition towards sustainable practices. Another study \citep{mourad2016recycling} highlights the three phases of handling surplus food waste, namely, prevention, recovery, and recycling. This takes the utilization of surplus food waste to beyond human consumption. As already stated, the results of ICT application to the various stages of food wastage has been studied \citep{ciaghi2016beyond}. ICT-based solutions have also been identified \citep{svenfelt2016sustainable} for efficiency through monitoring and assessment of environmental impact, enhanced transparency and traceability in the food system, creation of a network between actors in the food chains, and to influence and change food practices. A study \citep{garrone2016reducing} explores multiple case studies to provide available surplus management options and the factors that make these options attractive and applicable. Another study \citep{gram2016food} finds out and suggests certain recommendations for the improvement of food redistribution. A research \citep{persson2016food} contrasts a food redistribution environment with and without a food bank into the picture and evaluates the effects of a food bank on the environment and the finance of the participants of the activity. Another research \citep{anselmo2016re} studies \textit{Re-Food} (\url{https://www.re-food.org/en}), a Portuguese food redistributor, and finds out that logistic related issues stop it from fully eliminating food waste. It also notes that the firm brings together people from different stages of the redistribution activity to adhere to a common cause that is the local elimination of food waste and hunger issues. There is research \citep{orgut2016modeling} on mathematical models of the food circulation by food banks to analyze and optimize the effective and equitable distribution of food, i.e., food distributed to each service area should be proportional to the demand of food in that service area.

One study \citep{vlaholias2015charity} puts forward the concept of food redistribution and the principles that follow it, and analyse the food redistribution activities on that basis. A Norwegian study \citep{capodistrias2015reducing} takes the solution of food waste generation to beyond human redistribution and consumption. It includes approaches like fodder for animal, biogas generation, and even compost, stretching the utility of surplus beyond the human edible window. One study \citep{silvennoinen2015food} carried out in Finland calculates the number of cooked meals at up to 10,000 a year, and that of redistributed food bags as up to 270,000 bags a year by one organization, from surplus donated or redistributed food. Another study \citep{neff2015reducing} analyzes government policies with respect to handling food waste generation and overall people welfare and categorizes them into policies those help achieve a relationship between both the above goals and those which help achieve one of them while degrading the situation for the other.

Food rescue has been highlighted \citep{lindberg2014food} as an initiative in the emergency food sector internationally as an attempt to reduce food waste and to improve supplies to providers and consumers. A study \citep{caraher2014surplus} brings into light the steps through which food producers recover and donate their surplus food for redistribution to the food banks. It provides valuable insight on the visceral operations of food organizations that relate to their tax savings, waste management, and society rapport. While accepting surplus food redistribution as a key solution to food waste reduction, another study \citep{garrone2014opening} takes one step ahead by also suggesting a model for surplus food generation and management. A study \citep{downing2014food} notes that although food banks are a successful concept in the food management process, the government does not track its usage. Research \citep{kim2014food} exists on the food redistribution systems that existed in China at around 200 BC. It notes the redistribution by the emperor to the commoners in need and noted a bias towards saving physically capable farmers rather than addressing to the needs of the economically endangered ones. It also sheds light on an exceptional method of redistribution used by the rulers to feed people of inferior status through feast leftovers. Another study \citep{giuseppe2014economic} proposes a model to maximize the economical benefits of food redistribution for the retail donor organizations. The model suggests optimal time for withdrawing food items from shelves for redistribution and also donation quantities for human and livestock consumption such that retailer profit is maximized.

\subsection{Food Sharing Communities}

To understand the surplus redistribution frameworks in operation, we need to analyze the models used by various food sharing communities. According to \textit{Food Tank} (\url{https://foodtank.com/}) \citep{foodtank}, there are several such food sharing communities that are already in operation globally. Some of the leading names amongst them include \textit{412 Food Rescue} (\url{http://412foodrescue.org/}; courtesy: \textit{Ariel Procaccia} and group) operating in Pittsburgh, Pennsylvania, United States; \textit{Copia} (\url{https://www.gocopia.com/}) operating in San Francisco, California, United States; \textit{Community Food Rescue} (\url{https://communityfoodrescue.org/}) operating in Montgomery County, Maryland, United States; \textit{Food Cowboy} (\url{http://www.foodcowboy.com/}) operating in Bethesda, Maryland, United States and \textit{FoodLoop} (\url{https://www.foodloop.net/en/}) operating in Cologne, Germany. In India, a developing country, some notable presences include \textit{Feeding India} (\url{https://www.feedingindia.org}), operating in Kolkata, West Bengal, and 16 other cities in other states; \textit{Robin Hood Army} (\url{https://robinhoodarmy.com}), operating in Kolkata, West Bengal, and 8 other cities in other states; \textit{Roti Bank by Dabbawalas} (\url{https://rotibankfoundation.org}), operating in Mumbai, Maharashtra; \textit{Mera Parivar} (\url{https://www.meraparivar.org/}), operating in Gurugram, Haryana; and \textit{Shelter Don Bosco} (\url{www.shelterdonbosco.org}), operating in Mumbai, Maharashtra.

Although, some communities use ICT platforms to match donors and receivers (like \textit{Copia:} \url{https://www.gocopia.com/}), most others use ICT in much simpler forms like apps for volunteer organization or social media for receivers to connect with the donors. This lack of complexity of the usage of the ICT platform renders available a wide gap of possibilities that can be addressed to smoothen the surplus food redistribution process further. With a more complex use of ICT, the communities can go across borders, redistribute faster, communicate seamlessly, reach a wider audience, track food quality and movement, and exploit higher scalability, just to name a few. To sum up, these, and several other such, communities have already been using ICT-mediated methods successfully to move surplus food from parties, organizations, individuals etc. to people with food deficit. Although, when it comes to developing countries like India, as the facts point out, their presence in several of the smart cities, not including the upcoming smart ones, is very limited.

\subsection{ICT Supporting Research}

ICT, with time, is gaining popularity in bridging the gap between the surplus and the deficit. It is also evident that ICT is playing an important role as a mediator in achieving the goal of a sustainable food future. As stated previously, a study \citep{davies2018re} points out that food is wasted at various steps while it reaches the ultimate consumers. However, with the intervention of ICT, the majority of the surplus recovery happened at the retail stage. \textit{Foodsharing} (\url{https://foodsharing.de/}) pitches food sharing as a charitable practice. However, modern food sharing is mostly fueled from social and economic standpoints, often mediated by ICT. As mentioned, another research \citep{ganglbauer2012creating} conducted in the US, pointed out that maximum food wastage occurs in the households. The use of ICT can attack this problem by attempting behavioural modification through the culture of good, long term food handling practices. ICT is pointed out as valuable tool capable of engaging consumers in cultivating better practices while consuming less of their time. Results of ICT application to the various stages of food wastage, as mentioned before, are documented as well \citep{ciaghi2016beyond}. The popularity of ICT platforms - web, social media and handheld device apps - in redistributing surplus food with leading initiatives like \textit{Food Runners}, \textit{Copia}, \textit{RePlate} and \textit{Food Recovery Network} that exploit the benefits of the technology platforms, as already stated \citep{weymes2019re}, reveals our gaping scope of research. It was also observed that using technology for food redistribution attracted better quality surplus food and increased the scalability of the operation. It noted that redistribution, without the use of ICT, enforced limits like communication gaps, distribution delays and physical distances, on the operation of communities. The use of ICT led to higher satisfaction in the redistribution process, greater recoverability of food, lesser landfills, faster movement of food, better scalability and an overall better quality of the surplus.

\section{Method} \label{Sec-3}

An app will be used as an ICT-based mediator for surplus food distribution. On startup, the app will have options for registration and login. Then it will allow Users (hereafter interchangeably referred to as \textit{Agents}) to raise food requests as donors, receivers or volunteers. The app will take other required information like location, surplus food types, their expiry dates (packaged food items), preparation times (cooked food items), pickup, drop, donation, receipt related information, etc. as input as well. Just before the request submission, the agent will be able to set preferences for donation, pickup, or volunteering, as applicable. Under the hood, the app will be using double-sided market based game-theoretic algorithms to map donors with receivers and volunteers, if available, taking into account their preferences and location settings along with the timing windows for which they have opted.

\section{Notations and Problem Formulation} \label{Sec-4}

\subsection{The Agents Involved}

Our model has three main agents:
\begin{itemize}
    \item The $d$ \textit{Donors}, $\mathbb{D} = \{\mathbb{D}_1, \mathbb{D}_2, .., \mathbb{D}_d\}$,
    \item The $r$ \textit{Receivers}, $\mathbb{R} = \{\mathbb{R}_1, \mathbb{R}_2, .., \mathbb{R}_r\}$,
    \item The $v$ \textit{Volunteers}, $\mathbb{V} = \{\mathbb{V}_1, \mathbb{V}_2, .., \mathbb{V}_v\}$.
\end{itemize}

All agents mentioned here, and henceforth, are representatives of the agents' food donation/requirement/transportation requests. One agent can have multiple requests represented by different Request IDs at the same time as we will see later in the paper. Here, nothing is assumed regarding the relation between the number of donors ($d$), that of the receivers ($r$), and that of the volunteer requests ($v$). However, under the current scenarios prevailing in developing countries like India, it is fair to assume that $d<r$. Also, the number of volunteer requests ($v$) is expected to fluctuate wildly depending on several real-time parameters.

\subsection{The Information Gathered}

Each agent submits/allows, as applicable, the details presented in \hyperref[Tab-1]{\textbf{Table 1. Information Gathered}} at the time of registration/app usage. All the three agents are abstracted into three classes and an object is instantiated for each agent involved. All the above details of each agent are accessible via functions defined in the respective object classes. Details on each piece of information have been discussed later in this paper. This information collected from the agents are only displayed to other agents involved in a match generated by the algorithm. This means that donors can know the details of the receiver who will receive the donor's donation and the volunteer who will be doing the transportation. This goes all ways, that is, receivers also can view the donor and the volunteer details upon a match produced by the algorithm, as can the volunteers view the other two agents' details.

\begin{table}
    \tbl{Information Gathered}
    {
        \begin{tabular}{lll} 
            \toprule
            \multicolumn{1}{c}{\textbf{Donors}} & \multicolumn{1}{c}{\textbf{Receivers}} & \multicolumn{1}{c}{\textbf{Volunteers}}\\
            \midrule
            Location & Location & Route\\ 
            Type of food to donate & Type of food required & Type of transport (motored or not)\\
            Amount of food to donate & Amount of food required & Payload capacity\\
            Packaging of the food items & \multicolumn{1}{ c }{-} & Air-conditioning status of the transport\\
            Preparation/Expiry time for food & \multicolumn{1}{ c }{-} & \multicolumn{1}{ c }{-}\\
            Image of the food to be donated & \multicolumn{1}{ c }{-} & \multicolumn{1}{ c }{-}\\
            Pickup date and time range & Requirement date and time range & Availability date and time range\\
            Preferred receivers, if any & Preferred donors, if any & Receivers, if any\\
            \bottomrule
        \end{tabular}
    }
    \label{Tab-1}
\end{table}

\subsection{The Data Structures Involved}

There will be a list $\mathbb{P}$, maintained at the server-end, defining types of food as belonging to \textit{Perishable} or \textit{Non-Perishable} categories, with mixed food type donors/receivers defaulted to the \textit{Perishable} category. Another list $\mathbb{C}$, also maintained at the server-end, grows by adding to itself new agents who have active/unmatched requests (not just logged into it). A list $\mathbb{M}$, created by the mechanism itself, is to store and update the matching outcomes of the mechanism and display the same to the involved agents. Classified on the basis of the \textit{type of food} field detail of the donors and the receivers, four lists further bifurcate the matching task into those of perishable and non-perishable food items. This is shown in \hyperref[Tab-2]{\textbf{Table 2. Bifurcation of Food Requests}}. All these lists are used by a centralized website/server to carry out the matching process based on the agent details. It then displays the matched details to the agents involved in the respective matching.

\begin{table}
    \tbl{Bifurcation of Food Requests}
    {
        \begin{tabular}{m{0.23\textwidth}m{0.23\textwidth}m{0.24\textwidth}m{0.24\textwidth}} 
            \toprule
            \multicolumn{4}{c}{\textbf{All Food Requests}}\\
            \midrule
            \multicolumn{2}{c}{\textbf{Perishable}} & \multicolumn{2}{c}{\textbf{Non-Perishable}}\\ 
            \midrule
            \multicolumn{1}{c}{\textbf{Donations}} & \multicolumn{1}{c}{\textbf{Receipts}} & \multicolumn{1}{c}{\textbf{Donations}} & \multicolumn{1}{c}{\textbf{Receipts}}\\
            \midrule
            The Perishable Food Donors ($\mathbb{PFD}$) & The Perishable Food Receivers ($\mathbb{PFR}$) & The Non-Perishable Food Donors ($\mathbb{NPFD}$) & The Non-Perishable Food Receivers ($\mathbb{NPFR}$)\\ 
            \bottomrule
        \end{tabular}
    }
    \label{Tab-2}
\end{table}

\subsection{The Problem Formulation} \label{ProbForm}

We have $d$ donation requests $\mathbb{D}=\{\mathbb{D}_1, \mathbb{D}_2,..,\mathbb{D}_d\}$ to donate surplus food, and $r$ requirement requests $\mathbb{R}=\{\mathbb{R}_1, \mathbb{R}_2,..,\mathbb{R}_r\}$ to receive these donations. Also, $v$ volunteer requests $\mathbb{V}=\{\mathbb{V}_1, \mathbb{V}_2,..,\mathbb{V}_v\}$ are to mediate the transportation of these donations. The challenge here is to suggest a proper matching between these requests to make this process execute with minimum manual intervention. The simplified process to achieve the same is shown in \hyperref[Fig-1]{\textbf{Figure 2. Process Flow}}.

\begin{figure}
    \centering
    \includegraphics[width=\textwidth]{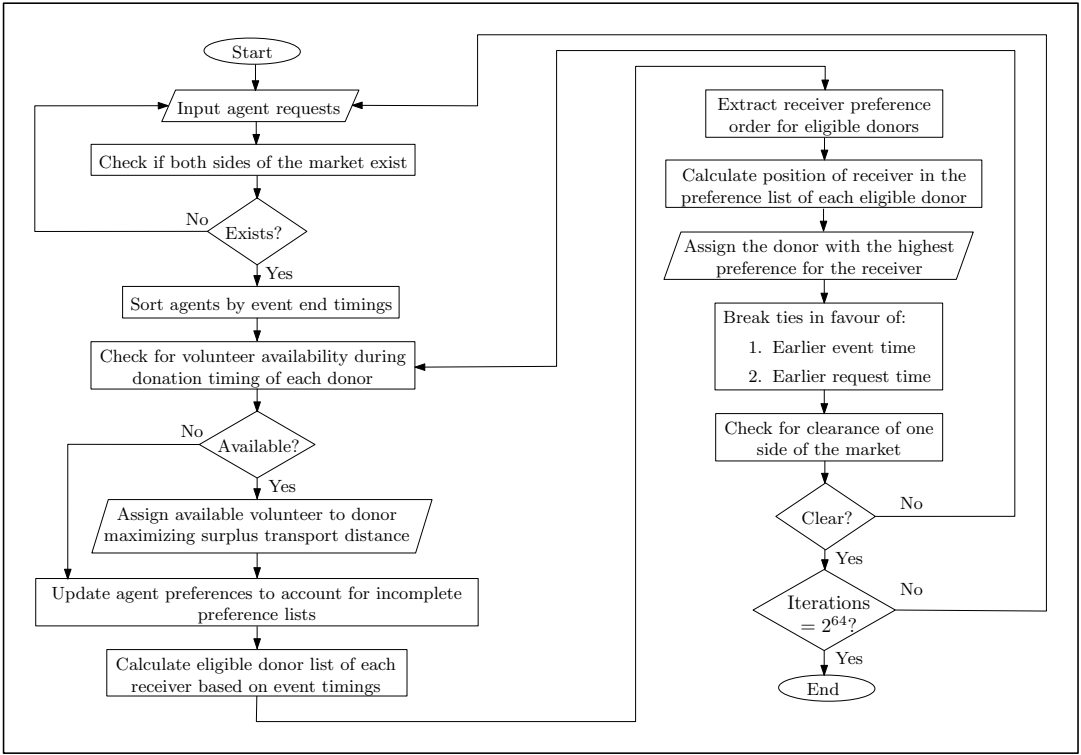}
    \caption{Process Flow}
    \label{Fig-1}
\end{figure}

Certain assumptions and constraints have been imposed on the system to make it simple. When tweaked, they change the environment of work for the model. These assumptions and constraints frequently take the face of a few thresholds or parameters used throughout the algorithm and often represent manageable, minimum, or maximum acceptable quantities for the system. These are defined as follows:

\begin{itemize}
    \item $T_o$ minutes is the minimum overlap time between a donor and a volunteer for the later to be assigned to the former. This threshold provides the volunteer time to reach the donor and can be set based on the traffic conditions of the city of operation.
    \item $T_l\%$ is the off routing threshold for the volunteers. In simpler terms, it is expected that volunteers can manage going off their route by $T_l\%$ of their travel distance to address the transportation requests of donors. It can be argued that the maximum total off-routing percentage ($\Gamma\%$) for one meal transportation under this setting will always be less than equal to $4 \times T_l\%$ of the volunteer route distance (proof included in Section \ref{Analysis}).
    \item $T_m$ grams is the threshold for a healthy meal size. This ensures proper nutrition of the meal consumer.
    \item $T_a\%$ of the donation weight is the threshold extra payload capacity of any volunteer's transportation over and above the capacity of the donation weight to be able to comfortably transport the donation with its packaging. This extra capacity is to allow for a smooth transportation of the food items without causing damage to them.
    \item $T_P^{nm}$ kilometers is the threshold distance that \textit{perishable food} can travel without spoiling in a \textit{non-AC and non-motored} transportation when no volunteer has been assigned for the transportation.
    \item $T_P^m$ kilometers is the threshold distance that \textit{perishable food} can travel without spoiling in a \textit{non-AC and motored} transportation when no volunteer has been assigned for the transportation.
    \item $T_{NP}$ kilometers is the threshold distance for \textit{non-perishable food} to travel in a non-motored transportation in case no volunteer has been assigned for the transportation.
    \item $T_d$ minutes is the threshold time before the donation start and end timings during which the donation request is considered available for matching and transportation assignment. This threshold is to account for real-time delays in the matching process and the volunteer arrival at the donor locations.
    \item $T_r$ minutes is the threshold time before the requirement start and end timings during which the requirement is considered available for donation search process. This threshold is to account for real-time delays in the matching process, the food donations from the donors, and transportation of the same by the volunteers to the receiver locations.
    \item $T_w$ minutes is the threshold beyond which a displayed match will be automatically cancelled when kept unaccepted by any of the involved agents and all requests involved in it will be considered for the next matching iteration.
\end{itemize}

\noindent All requests are raised by the agents ($\mathbb{D} \cup \mathbb{R} \cup \mathbb{V}$) in the mobile app and are first queued into a list $\mathbb{C}$ for processing. The variables (agent details) extracted for the matching process are mentioned in \hyperref[Tab-3]{\textbf{Table 3. Variables Extracted for Processing}}. To facilitate the matching of perishable food requests before the non-perishable ones, all food requests are classified as already in \hyperref[Tab-2]{\textbf{Table 2. Bifurcation of Food Requests}}.

\begin{table}
    \tbl{Variables Extracted for Processing}
    {
        \begin{tabular}{m{0.28\textwidth}m{0.32\textwidth}m{0.36\textwidth}} 
            \toprule
            \multicolumn{1}{c}{\textbf{Donors}} & \multicolumn{1}{c}{\textbf{Receivers}} & \multicolumn{1}{c}{\textbf{Volunteers}}\\
            \midrule
            Location & Location & Route\\ 
            Type of food to donate & Type of food required & Type of transport (motored or not)\\
            Amount of food to donate & Amount of food required & Payload capacity\\
            \multicolumn{1}{c}{-} & \multicolumn{1}{c}{-} & Air-conditioning status of the transport\\
            Pickup date and time range & Requirement date and time range & Availability date and time range\\
            Preferred receivers, if any & Preferred donors, if any & Receivers, if any\\
            \bottomrule
        \end{tabular}
    }
    \label{Tab-3}
\end{table}

\begin{figure}
    \centering
    \includegraphics[width=\textwidth]{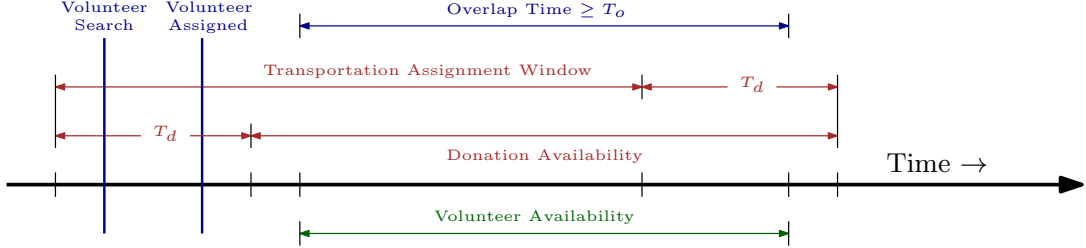}
    \caption{Volunteer Assignment Chronology}
    \label{Fig-2}
\end{figure}

\noindent Matching of food items are done internal to these groups. An additional volunteer group, which is involved in the transportation of both perishable and non-perishable food items, is also present to classify the volunteer requests. Volunteer assignment for each donation request is done $T_d$ minutes before donation availability and with donor-volunteer pairs having at least $T_o$ minutes of availability overlap to account for real-time delays as shown in \hyperref[Fig-2]{\textbf{Figure 3. Volunteer Assignment Chronology}}. Similarly, receivers are considered for matching $T_r$ minutes before their requirement start time.

\begin{figure}
    \centering
    \includegraphics[width=\textwidth]{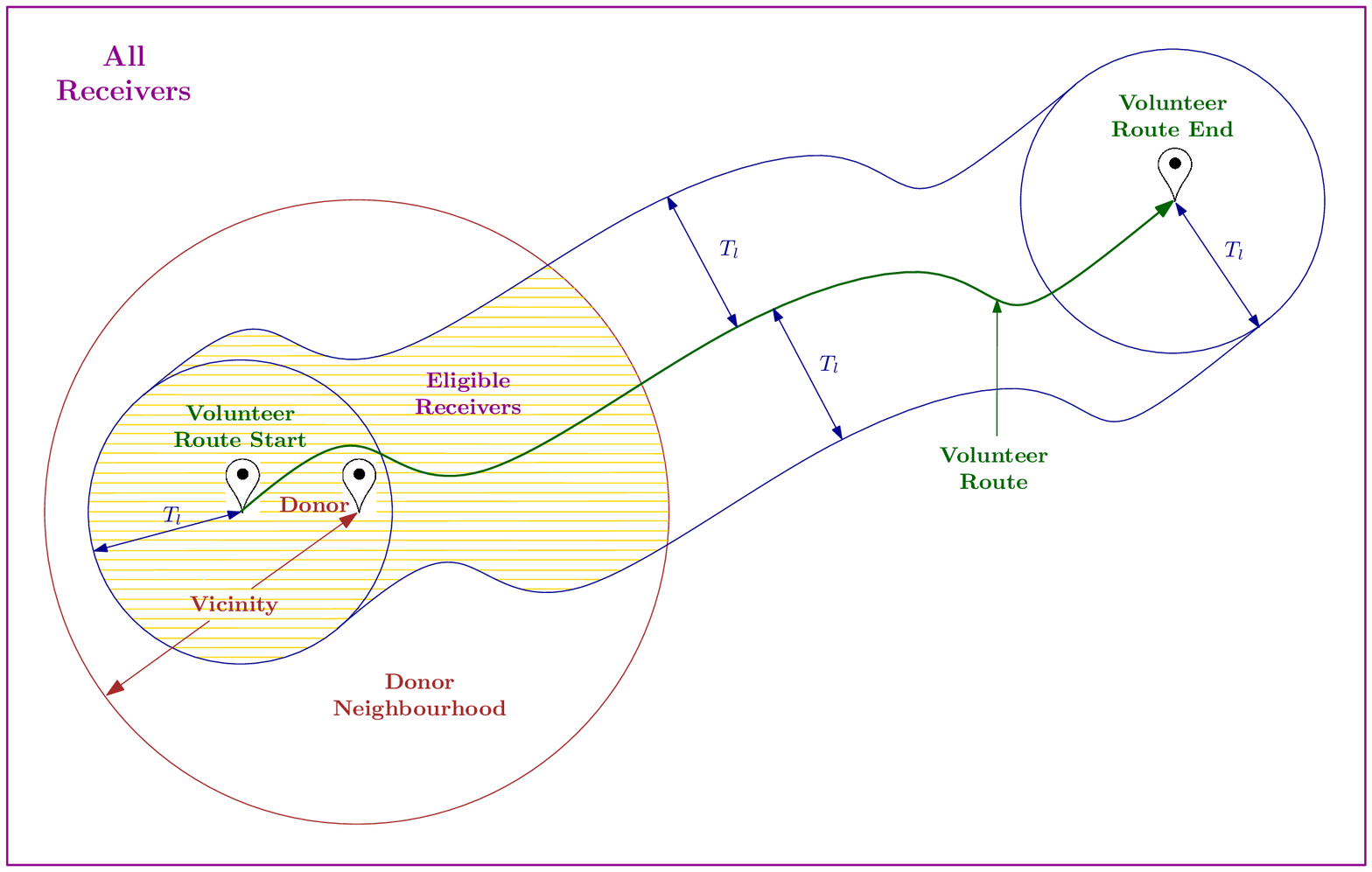}
    \caption{Receiver Eligibility Through Priority Modification}
    \label{Fig-3}
\end{figure}

Let us now follow a donation request from a donor to a receiver via a volunteer. Say, the $i^{th}$ donor wants to donate some food from the $j^{th}$ location, and is offering food of the $k^{th}$ type. The ICT-based platform splits this request into several requests of meal size $T_m$. Let us, for the time being, track the movement of this donor's $l^{th}$ meal. Therefore, we are now effectively tracking the $l^{th}$ meal, of the $k^{th}$ food type, donated by the $i^{th}$ donor, from the $j^{th}$ location. Let this donation request be our virtual donor $\mathbb{D}_{ijkl}$. We are now in search of a volunteer for this donation. So, we now shortlist all volunteers in this location who have a minimum $T_o$ minutes of availability overlap with this virtual donor. We narrow down this list further by selecting only those volunteers who have $\mathbb{D}_{ijkl}$ in a radius of $T_l\%$ of their travel distance, from their start location. This is shown in \hyperref[Fig-3]{\textbf{Figure 4. Receiver Eligibility Through Priority Modification}}. We choose the volunteer $\mathbb{V}_m$ who maximizes the distance to which this donation $\mathbb{D}_{ijkl}$ can be transported subject to certain restrictions. This distance $Vicinity$, that determines the donor neighbourhood, is lower bound by $T_P$ for perishable food, and $T_{NP}$ for non-perishable ones. Volunteers are only considered if the payload capacity of the volunteer is $T_a\%$ more than the donation weight of $\mathbb{D}_{ijkl}$. The AC status of the transportation and whether the vehicle is motored or not are also factors in choosing the volunteer. This determines how far the surplus can be sent for consumption before it spoils. In case we did not have any volunteer to assign to this donor, we would still proceed to receiver assignment process bypassing the volunteer assignment step. However, in that case, the $Vicinity$ would be scaled down to its lower bound ($T_P^{nm}/T_P^m/T_{NP}$) for the type of food being donated and whether the transport is motored or not, to make the transportation feasible for the agents. If the agents can predict volunteer unavailability in advance, they can themselves submit a volunteering request, in parallel, to aid the transportation themselves. This volunteering request then needs to have the receiver as its only entry in the volunteer's receiver list.

\begin{figure}
    \centering
    \includegraphics[width=\textwidth]{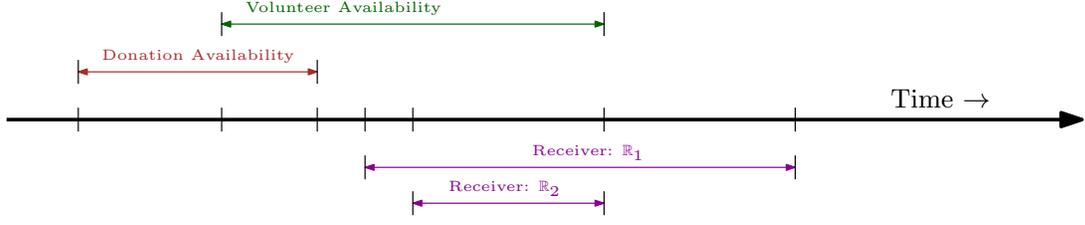}
    \caption{End Time Based Sorting for Receivers}
    \label{Fig-6}
\end{figure}

\begin{figure}
    \centering
    \includegraphics[width=\textwidth]{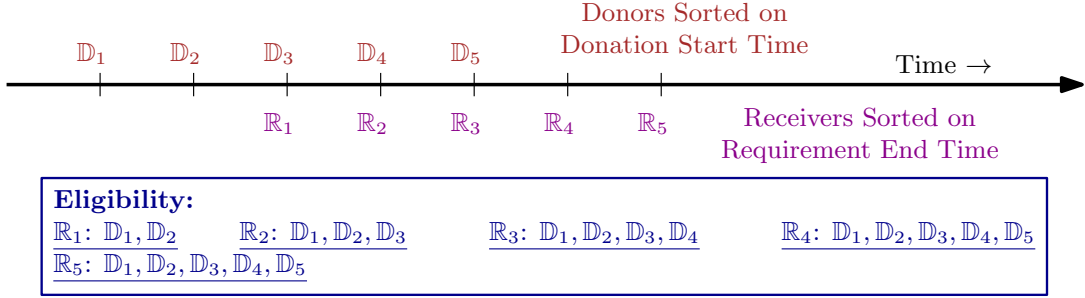}
    \caption{Donor Eligibility per Receiver}
    \label{Fig-7}
\end{figure}

We now have our volunteer $\mathbb{V}_m$ assigned to the virtual donor $\mathbb{D}_{ijkl}$ and ready for receiver assignment. Let us now observe \hyperref[Fig-6]{\textbf{Figure 5. End Time Based Sorting for Receivers}}, wherein one donor wants to donate two meals and two receivers want to receive one meal each. We assume volunteer availability for this donation. The two receivers have timing such that the donation is available before their requirement start times. However, the receiver with a later requirement start time ($\mathbb{R}_2$) has a very short availability. We can now address requirement requests starting with the earliest requirement \textit{start} time or the earliest requirement \textit{end} time. Assuming that we start with the earliest requirement \textit{start} time, by the time the volunteer delivers a meal to $\mathbb{R}_1$ who has the earlier requirement start time, $\mathbb{R}_2$'s requirement might end. Instead, if we start addressing the requirement requests starting with the earliest requirement \textit{end} time, we can cater to the request of $\mathbb{R}_2$ first, followed by $\mathbb{R}_1$ who has a longer request availability time, thereby attending to both the receivers. Thus, we opt for addressing the requirements in an \textit{earliest end time first} fashion. Receivers can only receive food that has already been donated, that is, the donation end time of a donor has to be before the requirement end time of a receiver. \hyperref[Fig-7]{\textbf{Figure 6. Donor Eligibility per Receiver}} clearly depicts the same and also states the donors eligible for each receiver. The priority order of the eligible donors is extracted from the receiver's preference list. For donors not in the receiver's priority list, they are assigned the minimum priority each as in \hyperref[Tab-4]{\textbf{Table 4. Priority Extraction and Augmentation}} for receiver $\mathbb{R}_n$.

\begin{table}
    \tbl{Priority Extraction and Augmentation}
    {
        \begin{tabular}{lll} 
            \toprule
            \multicolumn{1}{c}{\textbf{Eligibility List}} & \multicolumn{1}{c}{\textbf{Original Preference}} & \multicolumn{1}{c}{\textbf{Extraction \& Augmentation}}\\
            \midrule
            \multicolumn{1}{c}{$\underline{\mathbb{R}_n}$: $\mathbb{D}_p, \mathbb{D}_q, \mathbb{D}_r, \mathbb{D}_s$} & \multicolumn{1}{c}{$\underline{\mathbb{R}_n}$: $\mathbb{D}_q \succ \mathbb{D}_t \succ \mathbb{D}_p$} & \multicolumn{1}{c}{$\underline{\mathbb{R}_n}$: $\mathbb{D}_q \succ \mathbb{D}_p[ \succ \mathbb{D}_r=\mathbb{D}_s]$}\\
            \bottomrule
        \end{tabular}
    }
    \label{Tab-4}
\end{table}

\begin{table}
    \tbl{Donor to Receiver Assignment}
    {
        \begin{tabular}{lll} 
            \toprule
            \multicolumn{3}{c}{\textbf{Donor Preferences}}\\
            \midrule
            \multicolumn{3}{c}{$\mathbb{D}_p$: $\mathbb{R}_1 \succ \mathbb{R}_2 \succ \mathbb{R}_3 \succ \mathbb{R}_4 \succ \underline{\mathbb{R}_n} \succ ..$ \hspace{2cm} $\mathbb{D}_q$: $\mathbb{R}_1 \succ \underline{\mathbb{R}_n} \succ \mathbb{R}_3 \succ \mathbb{R}_4 \succ \mathbb{R}_5 \succ ..$}\\
            \multicolumn{3}{c}{$\mathbb{D}_r$: $\mathbb{R}_1 \succ \underline{\mathbb{R}_n} \succ \mathbb{R}_3 \succ \mathbb{R}_4 \succ \mathbb{R}_5 \succ ..$ \hspace{2cm} $\mathbb{D}_s$: $\mathbb{R}_1 \succ \mathbb{R}_2 \succ \mathbb{R}_3 \succ \underline{\mathbb{R}_n} \succ \mathbb{R}_5 \succ ..$}\\
            \midrule
            \multicolumn{1}{c}{\textbf{Updated Preferences}} & \multicolumn{1}{c}{\textbf{Preference Positions}} & \multicolumn{1}{c}{\textbf{Best Match}}\\
            \midrule
            \multicolumn{1}{c}{$\underline{\mathbb{R}_n}$: $\mathbb{D}_q \succ \mathbb{D}_p[ \succ \mathbb{D}_r=\mathbb{D}_s]$} & \multicolumn{1}{c}{$\underline{\mathbb{R}_n}$: $\mathbb{D}_p.5, \mathbb{D}_q.2, \mathbb{D}_r.2, \mathbb{D}_s.4$} & \multicolumn{1}{c}{$\mathbb{R}_n \Longleftrightarrow \mathbb{D}_q$}\\ 
            \bottomrule
        \end{tabular}
    }
    \label{Tab-5}
\end{table}

Let us now get back to our virtual donor $\mathbb{D}_{ijkl}$ to whom volunteer $\mathbb{V}_m$ was assigned and was ready for receiver assignment. We now start the process from the receivers' end and reach this donor of ours as a match. Let, for our tracking purposes, our donor $\mathbb{D}_{ijkl}$ be matched to the receiver $\mathbb{R}_n$, in the future. Then $\mathbb{R}_n$ surely has $\mathbb{D}_{ijkl}$ in its extracted and augmented preference list. But to be matched to each other, $\mathbb{D}_{ijkl}$ also needs to have the highest preference for $\mathbb{R}_n$, compared to all other eligible donors for this receiver, in its extracted and augmented preference list prepared. This list is prepared, similar to what was done for the receiver preference (\hyperref[Tab-4]{\textbf{Table 4. Priority Extraction and Augmentation}}), from the donor's submitted preference list and the receivers in the donor's neighbourhood determined by the calculated $Vicinity$. The assignment process is shown with an example in \hyperref[Tab-5]{\textbf{Table 5. Donor to Receiver Assignment}} for the receiver $\mathbb{R}_n$. In case of ties in preferences as introduced by the augmentation of priorities, the agent with the earliest request submit time is favoured. In addition to that, for the facilitation of the transportation, $\mathbb{R}_n$ also needs to be in a location that is at a maximum perpendicular distance of $T_l\%$ of the assigned volunteer $\mathbb{V}_m$'s route distance, from the route itself. This is as shown in \hyperref[Fig-3]{\textbf{Figure 4. Receiver Eligibility Through Priority Modification}}. Post matching, food transportation from the donor $\mathbb{D}_{ijkl}$ to the receiver $\mathbb{R}_n$ via the volunteer $\mathbb{V}_m$ is as shown in \hyperref[Fig-4]{\textbf{Figure 7. Food Movement Through Volunteer}}. However, if any of the agents involved in a match reject the same, or do not accept it within $T_w$ minutes of the match generation, the match is cancelled and all agent requests involved are rescheduled for matching.

\begin{figure}
    \centering
    \includegraphics[width=\textwidth, height=7cm]{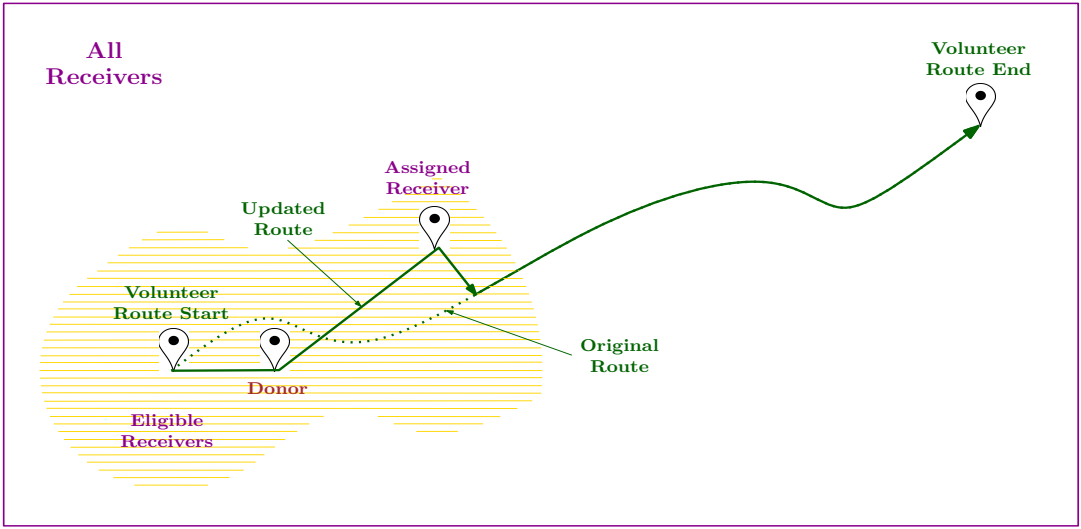}
    \caption{Food Movement Through Volunteer}
    \label{Fig-4}
\end{figure}

\section{The Model Operations}

To understand how the model operates, we will warm-up with the simplest possible view of it and then gradually up-shift gears towards more complex views of the same. An overview of the model is already presented in \hyperref[Fig-5]{\textbf{Figure 1. Overview of the Model}}.

\subsection{Registration and Login}

Before we get into the analysis of complex scenarios, let us first handle the registration and the login processes. As an agent registers with the agent's personal details, contact information, email information, and a password, a specific Agent ID gets generated which is essentially a sequence number, and is to be used for logging into the app.

\subsection{Introducing the App}

Let us now log into the app and take a tour of the same via the different possible routes that an agent can take. While, for volunteer activity, a start and an end location for route determination is prompted, for donors and receivers only one location is accepted on app login. Locations can either be provided using the GPS or manually. Hereafter, different agents are directed through markedly different routes while using the app, although the location acceptance step had already started to differentiate between the volunteers and the other agents. Volunteers will be prompted for the type of transport, whether motored or not, the air-conditioning status of this transport, and the payload capacity that can be transported by them. These three details are later used to calculate the maximum distance to which food can be transported without spoiling it, especially for perishable food. Availability date and time range is also requested which defaults to the current date and time, if left missing. At last, an optional prompt for choosing receivers to deliver the food is provided. This is helpful for those donor-receiver matching pairs who do not receive any volunteer due to availability issues. If this choice is taken by the volunteers, and this receivers' list is not empty then these volunteers only receive those transportation requests that involve the receivers in this volunteer's receivers' list. Otherwise, when this list is empty, all receivers are considered for this volunteer.

Next, we follow the receivers' route through the app. Post entering the location as described above, a receiver needs to select the type of food that is required, in terms of \textit{freshly cooked}, \textit{frozen uncooked}, \textit{frozen cooked}, \textit{packaged solid}, \textit{packaged liquid}, \textit{fresh produce}, \textit{fruits and vegetables}, etc. This is later on used for classification into perishable and non-perishable food requirements. The amount of food required is also fed in along with the requirement start and end date and time. For all calculations, the requirement start and end date and time are considered as the \textit{event start and end date and time}. At last, an optional prompt for choosing preferred donors is provided. When this list is empty, all donors are considered with equal priority. Otherwise, this list, along with the donors left out from this list with lowest priority assigned to them are considered at the time of generating a match.

Lastly, we have the donors. They, similar to the receivers, need to select the type of food that they want to donate along with the weight of the food to be donated. For this, the same list as that for the receivers is prompted to the donors. This is also used for classifying the donation as that of perishable or non-perishable food. However, donors get additional prompts to mention about the packaging of the food item and to upload an image of the food to be donated. The former relates to the container to carry cooked food items and is an essential factor in retaining high recoverability of the same, and the latter is meant for being a visual clue of the current condition of the surplus. Recoverability also depends on the preparation date and time, for cooked food items, or the expiry date, for uncooked ones, and these are the next details that the donors need to provide.

The time elapsed from the preparation date and time to the pickup time of the food, food cooked items, or the time that remains between the pickup date and the expiry date is a major determining factor for safe recovery and reuse of the surplus. Thus, the donors need to provide a date and time range for the pickup. For all calculations, the donation/pickup start and end date and time are considered as the \textit{event start and end date and time}. Similar to the receivers, the donors also receive an optional prompt for choosing their preferred receivers. This list of preferred receivers, if empty, defaults to the neighbourhood of the respective donors. Otherwise, it is considered after amalgamation with the donor's neighbourhood.

Now, for donors, the app chops up the donation amount in meal chunks of a threshold of $T_m$ grams and creates separate requests for the same donor by adding a sequence number to the donor identifier,  the arrival date and time, and the request date and time of the donor. This accounts for the difference in the requested amounts and the donated amounts. However, all this processing is abstracted from all the agents unless the involved agents in different matching triplets of the different requests of the same donor, are different. This is further reinforced by the app clubbing together different matching assignments into one, if the involved agents are the same. It is to be noted that \textit{event start/end date and time} for donors refer to the \textit{food donation start/end date and time}, and that for the receivers refer to the \textit{food requirement start/end date and time}. This is to mark the earliest availability of food for donation and the latest usability of the donated food at the extreme ends of the process. This information is later used to sort the donors and the receivers for matching so as to promote maximum surplus food movement.

All requests have a request identifier associated with it, and is formed by a request sequence number, initiated to 1 at agent registration, appended to the agent identifier which is also a sequence number assigned to the agent at registration. This request sequence number runs independently for each agent. After each matching iteration, the assignments made in each iteration of the FDRM-CA mechanism are displayed to the involved agents via deferred notifications in their respective devices. To facilitate quick request submissions, the app will cache the last used agent data for each type of request amongst donor, receiver, and volunteer requests, raised by the agent. These are later auto-filled in the input spaces when this agent chooses a similar path through the app. Of course, some fields are never auto-filled, like the image for surplus food donation.

\subsection{Basic Operation}

Let us start with the initial scenario where all the donors, the receivers and the volunteers are new to the app and submit no preferences. In this case, volunteers are assigned as per availability. Based on the type and air-conditioning status of the volunteer transport, the neighbourhood of each donor is computed in terms of receivers. Out of them, the receivers in the volunteer's route and $T_l\%$ of the route distance around the route are then treated as the donor's preference. Receiver preferences are not computed, and all donors donating before the requirement start time are considered with equal priority for the receivers having no donor preference. Then, the donors are matched with the receivers as per the computed donor preferences. A slightly different situation occurs when donors provide their preferences and receivers do not. The operation is identical to the previous one except that receivers in this preference list who are present in the neighbourhood, are prioritized over the ones who are not present. If receivers also provide preferences, then chronologically available donors present in this priority list get preferred over those absent in the list.

\subsection{Extended Model}

A more complicated situation arises when both donors and receivers, after using the app for some time, provide preferences of their own. It then resembles the extended mode of operation of the app. As before, volunteers are assigned as per availability, and based on the type and air-conditioning status of the volunteer transport, the neighbourhood for each donor, in terms of receivers, is calculated and then sorted by food requirement end time. Out of them, the receivers in the volunteer's route and $T_l\%$ of the route distance around the route are treated as the donor's preference. All donors donating before the requirement start time are considered for the receiver's preference list. All such eligible donors absent in the receiver's preference are appended to the same with lowest preference. At last, the donors are matched with the receivers as per the computed donor and receiver preferences.

It is practical to assume that volunteer availability will fluctuate wildly with respect to the time of the day, the day of the week, the season, the weather, the ongoing occasions, and such other parameters. Let us now analyse a case where a donor is not assigned a volunteer due to unavailability. This situation is handled in a two fold way. Firstly, the $Vicinity$ of this donor is set in such a way that the transportation distance becomes trivial. Secondly, if the receivers can predict, considering practical situations, that volunteers may be unavailable, they can enroll their own volunteers to do the transportation themselves. In this case, these volunteers provide only this receiver in their receivers' list. When donors mapped to these volunteers are assigned receivers, the volunteers' receiver list is checked before the assignment is made. This also holds true when the receivers are shown a mapping that has no volunteers assigned to it, making the receivers take charge of transportation of the surplus. Essentially, both donors and receivers can be volunteers, simultaneously, as shown in \hyperref[Fig-5]{\textbf{Figure 1. Overview of the Model}}. There will be scenarios where the amount that a donor wants to donate does not match with the amount that a receiver wants to receive or with the amount that a volunteer can transport. All these probable scenarios, with donation amounts greater than a threshold of $T_m$ grams, are handled by chopping down all donor requests into multiple requests, each with a meal size of $T_m$ grams. In case food items are not divisible perfectly into a $T_m$ grams meal size (for packaged food), the least number of units making the meal size over $T_m$ grams is acceptable as a meal. While doing this for each such donor, each meal is treated as a separate donation request and different sequence numbers are augmented at the end of the donor identifier to generate different request numbers for each such chopped request. This same sequence number is also added to both the food donation/requirement time and the arrival time of this donor request to uphold the tie breaking conditions.

The handling of these donations with amount more than $T_m$ grams is handled at the app's client end where it is chopped down to create multiple requests from the same donor using sequence numbers to distinguish them from each other. Similarly, when the agents (donor, receiver, and, if assigned, volunteer) are the same triplet for multiple matched requests, with differences only in the sequence numbers used to disambiguate the requests, they are clubbed and displayed as a single request to the agents. Therefore, the granularity of the slicing and dicing of agent requests are kept abstracted from the agents as much as possible until the involved agents are different for different matching requests. This abstraction is enforced both during the placement of requests by donors and during the display of the generated match to the involved agents. Thus, each donor can have multiple volunteers and receivers assigned to the donor. Receivers, more often than not, will have multiple donors and/or volunteers assigned to them. Note that receiver requests are not chopped into meals. Rather, these requests are treated as a single unit with an \textit{Amount} property along with a \textit{RemainingAmount} method that returns the unfulfilled requirement amount when called via a \textit{Receiver} class object. A similar method, \textit{RemainingPayload}, of the \textit{Volunteer} class returns the payload capacity that remains unused for a volunteer. Note that when multiple requests of the same donor are generated by the app as a result of request slicing and dicing, they are listed one after the other, making a volunteer, chosen as the optimal transporter for the previous request from this donor and still having unused payload capacity, be assigned to the same donor's successive pickup requests until the volunteer's payload capacity is exhausted. This takes us to the next case where some, or all, of the agents involved in a mapping, reject the same. The app will then automatically reject the complete matching and add the involved agents to the global list of active agent requests, $\mathbb{C}$, for matching in the upcoming rounds. The addition of these agents in the list ensures that the agent requests are made available again for proper processing in the next iteration of the mechanism.

\section{Mechanism} \label{Sec-5}

The Food Donor to Receiver Matching with Chronological Acceptance (FDRM-CA) mechanism is proposed for matching donors to receivers, if possible, via volunteers. This section is to discuss the mechanism and its sub-processes in detail in a breadth-first travel of the main algorithm.

\subsection{Detailed FDRM-CA}

The FDRM-CA mechanism actuates itself via the following three steps. These steps are repeated for a very large number of times to handle all incoming requests.

\begin{enumerate}
    \item Constant and parallel check for unaddressed agent requests using the New User Interrupt Routine (NUIR),
    \item Classification and TriFurcation of Users (CTFU) into Perishable, Non-Perishable, and Volunteer categories, and
    \item Chronological Acceptance, inspired by Roughgarden's lecture on \textit{Stable Matching} and Gale \& Shapley's algorithm \citep{gale2013college, roughgarden2016cs269i}, using Double Tie Breaking (CA-DTB), for the donor-receiver matching.
\end{enumerate}

\IncMargin{1em}
\begin{algorithm}
\label{Algo-1}
\DontPrintSemicolon
\SetAlgoVlined
\SetKwBlock{Parallel}{do parallel}{end}
\caption{FDRM-CA()}
    $\mathbb{C}$, $\mathbb{V}$, $\mathbb{PFD}$, $\mathbb{PFR}$, $\mathbb{NPFD}$, $\mathbb{NPFR} \leftarrow \phi$ \tcp*[f]{Initiate all lists}\\
    \Parallel(\tcp*[f]{Execute in parallel})
    {
        NUIR$(\mathbb{C})$ \tcp*[f]{Add unaddressed requests to the list}
    }
    \Parallel(\tcp*[f]{Execute in parallel})
    {
        \For(\tcp*[f]{Handle all requests in $n \rightarrow 2^{64}$ iterations}){$i \leftarrow 1$ \KwTo $n$}
        {
            $\mathbb{C}$, $\mathbb{V}$, $\mathbb{PFR}$, $\mathbb{NPFR}$, $\mathbb{PFD}$, $\mathbb{NPFD} \leftarrow$ CTFU$(\mathbb{C}$, $\mathbb{V}$, $\mathbb{PFD}$, $\mathbb{PFR}$, $\mathbb{NPFD}$, $\mathbb{NPFR})$ \tcp*[f]{Classify requests}\\
            $\mathbb{M}_P$, $\mathbb{PFD}$, $\mathbb{PFR}$, $\mathbb{V} \leftarrow$ CA-DTB$(\mathbb{PFD}$, $\mathbb{PFR}$, $\mathbb{V}$, $Food \leftarrow PERISHABLE)$ \tcp*[f]{Match perishable food}\\
            \tcp*[l]{Display each match to involved agents}
            \lForEach{$\mathbb{M}_j \in \mathbb{M}_P$}
            {
                $\mathbb{M}_j \rightarrow \mathbb{M}_j.MatchedAgents$
            }
            $\mathbb{M}_{NP}$, $\mathbb{NPFD}$, $\mathbb{NPFR}$, $\mathbb{V} \leftarrow$ CA-DTB$(\mathbb{NPFD}$, $\mathbb{NPFR}$, $\mathbb{V}$, $Food \leftarrow NON\_PERISHABLE)$ \tcp*[f]{Match non-perishable food}\\
            \tcp*[l]{Display each match to involved agents}
            \lForEach{$\mathbb{M}_j \in \mathbb{M}_{NP}$}
            {
                $\mathbb{M}_j \rightarrow \mathbb{M}_j.MatchedAgents$
            }
        }
    }
\end{algorithm}
\DecMargin{1em}

The detailed mechanism is as provided in \hyperref[Algo-1]{\textbf{Algorithm 1: The FDRM-CA Mechanism}}. After initializing all the lists (line 1), a parallel process for queuing unaddressed agent requests, the NUIR sub-process, starts (line 3). These lists are emptied of elements that have completed all processing on themselves, in each iteration of the mechanism, with the exception of agents who have not yet received a match. To avoid clutter, this emptying has not been explicitly stated in the algorithm. While the queuing of unattended agent requests continues for a very large number of iterations, it need not stop for the agents to be classified into Perishable, Non-Perishable and Volunteer groups (line 6), as it modifies a global list in parallel to the main mechanism program flow. The Perishable and Non-Perishable groups help prioritize the matching of perishable food items over the non-perishable ones, while the Volunteer group facilitates the movement of the surplus food from the donors to the receivers. This is followed by matching donors with receivers, if possible, via volunteers (lines 7 and 10), and then displaying each match to the involved agents (lines 9 and 12). The preferences of the donors and the receivers are taken into consideration while generating the matching. As already mentioned, the donors who wish to donate perishable food and receivers who wish to receive perishables, are the first ones to be matched.

It is worth mentioning that not all donors/receivers may get a match in their very first matching iteration through the mechanism. Some may be left over for the next iteration of this process destined for better matches as per their own preferences. As this process completes, the matches that were generated, are displayed to the involved agents (lines 9 and 12). Note that, to promote maximum recoverability of perishable food items, the generated matches for perishable food items ($\mathbb{M}_P$) is displayed (line 9) even before the matching for non-perishable food items can start (line 12). These match displays are non-blocking, that is, they do not stop the execution of the FDRM-CA mechanism to wait for the matched agents to accept the match. The classification and matching processes are repeated for a very large number of iterations. Each repetition involves new sets of donation, receipt, and volunteering requests or such unmatched requests from previous iterations. Note that these are not new sets of agents, since the same agents can have multiple requests on the same/different day(s). Instead, these are new sets of requests. Also, old and unmatched requests, as well as requests whose matches have been rejected by the involved agents, are forced to be processed again. The generated matches are displayed in a non-blocking way via device notifications.

\subsection{Detailed NUIR}

The detailed NUIR mechanism is as given in \hyperref[Algo-2]{\textbf{Algorithm 2: The NUIR Process}}. As $n \rightarrow 2^{64}$ (line 2), the algorithm captures all agent requests. This is a simple, interrupt driven mechanism to queue new agent requests for the subsequent iteration of classification and matching (lines 5 to 8), and re-queue old agent requests from unaccepted/rejected matches (lines 10 to 12). At each agent request submission, an interrupt ($NewUserRequestInterrupt$) is generated that triggers this routine to add the agent request to the list of active agent requests after processing for request splitting for multiple meals as per the donation amount (lines 5 to 8). Also, at the rejection of a matching by any of the involved agents or by partial (not by all agents involved in a match) acceptance by $T_w$ minutes of match generation, another interrupt ($MappingRejectionInterrupt$) is generated that triggers this routine to fetch all of the involved agents of this rejected mapping and add them to the same list (lines 10 to 12) for being matched in the next iteration of the matching process. This process runs in parallel with the main process without halting it, and modifies a global list $\mathbb{C}$ of active agent requests (line 1) for consumption in the next iteration of FDRM-CA. Thus, while the matching process is ongoing in the main program body, this process, in parallel, works on readying the requests that will be processed by the next iteration of the matching process.

\IncMargin{1em}
\begin{algorithm}
\label{Algo-2}
\DontPrintSemicolon
\SetAlgoVlined
\caption{NUIR$(\mathbb{C})$}
    \textbf{global} $\mathbb{C}$ \tcp*[f]{Agent requests available to other parallel processes}\\
    \For(\tcp*[f]{Handle all requests in $n \rightarrow 2^{64}$ iterations}){$i \leftarrow 1$ \KwTo $n$}
    {
        \If(\tcp*[f]{Check for interrupts from the app}){Interrupt}
        {
            \tcp*[l]{New agent request}
            \If{$Interrupt=NewUserRequestInterrupt$}
            {
                $\mathbb{A} \leftarrow Interrupt.GetRequest$\\
                \If{$\mathbb{A}.RequestType=\mathbb{D}.UserType \cup \mathbb{A}.Amount \geq 2 \times T_m$}
                {
                    $\mathbb{C}.AddRequest(\mathbb{A}.SplitIntoMeals(T_m))$ \tcp*[f]{Split into meals}
                }
            }
            \tcp*[l]{Unmatched agent requests}
            \If{$Interrupt=MappingRejectionInterrupt$}
            {
                \ForEach{$\mathbb{A} \in Interrupt.GetRequest$}
                {
                    $\mathbb{C}.AddRequest(\mathbb{A})$ \tcp*[f]{Add request to list} 
                }
            }
        }
    }
\end{algorithm}
\DecMargin{1em}

\subsection{Detailed CTFU}

The detailed mechanism is as follows. The CTFU process checks for perishable food item requests and classifies the agent request into perishable/non-perishable categories for donors/receivers (lines 2 to 9). It also classifies agents as volunteers (line 10). It updates five lists - $\mathbb{V}$, $\mathbb{PFR}$, $\mathbb{NPFR}$, $\mathbb{PFD}$, and $\mathbb{NPFD}$ - for consumption by the CA-DTB sub-process that follows it. Do note that these lists may have unconsumed agents from the previous round of CA-DTB, before the current round of CTFU may start. At the end, the algorithm returns the list $\mathbb{C}$ and the updated lists to the main body of the process (line 11) for consumption by the subsequent iterations of CA-DTB.

\IncMargin{1em}
\begin{algorithm}
\label{Algo-3}
\DontPrintSemicolon
\SetAlgoVlined
\caption{CTFU$(\mathbb{C}$, $\mathbb{V}$, $\mathbb{PFD}$, $\mathbb{PFR}$, $\mathbb{NPFD}$, $\mathbb{NPFR})$}
    \ForEach(\tcp*[f]{Classify all requests}){$\mathbb{A} \in \mathbb{C}$}
    {
        \If(\tcp*[f]{Receivers}){$\mathbb{A}.UserType = RECEIVER$}
        {
            \If(\tcp*[f]{Perishable food}){$\mathbb{A}.TypeOfFood = PERISHABLE$}
            {
                $\mathbb{PFR}.AddToList(\mathbb{A})$
            }
            \lElse(\tcp*[f]{Non-perishable food})
            {
                $\mathbb{NPFR}.AddToList(\mathbb{A})$
            }
        }
        \ElseIf(\tcp*[f]{Donors}){$\mathbb{A}.UserType = DONOR$}
        {
            \If(\tcp*[f]{Perishable food}){$\mathbb{A}.TypeOfFood = PERISHABLE$}
            {
                $\mathbb{PFD}.AddToList(\mathbb{A})$
            }
            \lElse(\tcp*[f]{Non-perishable food})
            {
                $\mathbb{NPFD}.AddToList(\mathbb{A})$
            }
        }
        \lElse(\tcp*[f]{Volunteers})
        {
            $\mathbb{V}.AddToList(\mathbb{A})$
        }
    }
    \Return $\mathbb{C}$, $\mathbb{V}$, $\mathbb{PFR}$, $\mathbb{NPFR}$, $\mathbb{PFD}$, $\mathbb{NPFD}$ \tcp*[f]{Return all lists}
\end{algorithm}
\DecMargin{1em}

\subsection{Detailed CA-DTB}

\hyperref[Algo-4]{\textbf{Algorithm 4: The CA-DTB Process}} works on the current agent requests, donors with donations to be available in $T_d$ minutes and receivers with requirements starting in $T_r$ minutes (lines 2 and 3). Next, it gets the volunteers available for each such donor $\mathbb{D}_i$ on the basis of donor-volunteer availability overlap time of at least $T_o$ minutes (line 5), volunteer transport type, and transport payload availability and its air-conditioning status (lines 10 to 15). Then, it matches the volunteer who maximizes the surplus transportation distance (lines 16 and 20). Next, receiver preferences are updated by extracting from it all chronologically available donors, appending donors with the minimum preference each, when required (line 23). This is followed by updating donor preferences by extracting and augmenting, as before, receivers from around the donor inside and on the circle with radius $Vicinity$ and around the assigned volunteer route (line 25). To determine the neighbourhood of each donor $\mathbb{D}_i$, on the basis of relevant receivers, the radius of the neighbourhood is calculated on the basis of a rough estimate of the spoiling time of the food items and the available transportation details. For perishable food items, this radius, the $Vicinity$, gets defaulted to $T_P^{nm}$ for non-motored, and to $T_{P}^m$ for motored transports. For non-perishable food items this $Vicinity$ gets a default value of $T_{NP}$. A volunteer de-route percentage threshold $T_l$ has also been used that symbolizes the percentage of the actual route distance of the volunteer that is a manageable diversion from the actual route of the volunteer for the surplus pickup. Post this, each receiver $\mathbb{R}_i$ is assigned to the eligible donor who has the highest preference for $\mathbb{R}_i$ (lines 26 to 31).

Ties are broken first with food donation start/requirement end timings, and next with request arrival timings. The last donation that is required to match the requirement, even if it overshoots the requirement amount slightly, is accepted anyway ($\sum_{j=1}^{m} \mathbb{D}_j.Amount - \epsilon = \mathbb{R}_i.Amount,\epsilon \rightarrow 0$). This step is repeated until one side of the market is cleared. At the end, the final matching is returned to the main body of the mechanism for displaying to the involved agents. Also, any unaddressed requests from this iteration of the mechanism are returned via their respective lists as well (line 32). Although, while breaking ties, there can be duplicates in the value of \textit{donation start/requirement end date and time}. However, to avoid duplicates in the values of \textit{request arrival date and time,} precautions have been taken. Firstly, they are not actual timestamps, but are sequence numbers generated by a sequencer maintained by the server. Also, these sequences maintain a \textit{total order} in the system and therefore cannot have duplicate values. This is essential in the functioning of the algorithm since this is the supreme tie breaking condition that is checked for donors/receivers having equal priorities for an agent and having the same \textit{donation start/requirement end date and time}. The generated match is stored in the list $\mathbb{M}$ in each iteration of the mechanism and is displayed to the involved agents.

\IncMargin{1em}
\begin{algorithm}
\label{Algo-4}
\DontPrintSemicolon
\SetAlgoVlined
\caption{CA-DTB($\mathbb{D}$, $\mathbb{R}$, $\mathbb{V}$, $Food$)}
    $\mathbb{M} \leftarrow \phi$ \hfill \tcp*[f]{Initialize the matching}\\
    $\mathbb{R}' \leftarrow \{\mathbb{R}_k \vert Now() \in [\mathbb{R}_k.RequirementStartDateTime - T_r$, $\mathbb{R}_k.RequirementEndDateTime - T_r]$, $\mathbb{R}_k \in \mathbb{R}\}$ \hfill \tcp*[f]{Current receivers}\\
    $\mathbb{D}' \leftarrow \{\mathbb{D}_k \vert Now() \in [\mathbb{D}_k.RequirementStartDateTime - T_d$, $\mathbb{D}_k.RequirementEndDateTime - T_d]$, $\mathbb{D}_k \in \mathbb{D}\}$ \hfill \tcp*[f]{Current donors}\\
    \ForEach(\tcp*[f]{Available volunteers}){$\mathbb{D}_i \in \{\mathbb{D}_l \vert \mathbb{D}_l \in \mathbb{D}', Now() \in [\mathbb{D}_l.DonationStartDateTime - T_d$, $\mathbb{D}_l.DonationEndDateTime - T_d]\}$}
    {
        $\mathbb{V}' \leftarrow$ \{$\mathbb{V}_j \vert \mathbb{V}_j \in \mathbb{V}$, $\mathbb{V}_j.RemainingPayload \geq (1 + T_a\%) \times \mathbb{D}_i.Amount$, $\mathbb{V}_j.Route.Start \pm T_l\% \times \mathbb{V}_j.Route.Distance \geq \mathbb{D}_i.Location$, $T_o \leq \mathbb{V}_j.AvailabilityDateTimeRange \cap \mathbb{D}_k.DonationDateTimeRange$, $(\mathbb{V}_j.Receivers = \phi) \cup (\mathbb{D}_i \in \mathbb{V}_j.Receivers)$\}\\
        \If(\tcp*[f]{Defaults for no available volunteer}){$\mathbb{V}' = \phi$}
        {
            \lIf{$Food = PERISHABLE$}
            {
                $\mathbb{D}_i.Vicinity \leftarrow T_P^{nm}$
            }
            \lElse
            {
                $\mathbb{D}_i.Vicinity \leftarrow T_{NP}$
            }
        }
        \Else(\tcp*[f]{Volunteers available})
        {
            \ForEach(\tcp*[f]{Optimize volunteer assignment}){$\mathbb{V}'_j \in \mathbb{V}'$}
            {
                \If{$Food \neq PERISHABLE \cup \mathbb{V}'_j.ACStatus = AC$}
                {
                    $Vicinity \leftarrow \mathbb{V}'_j.Route.Destination - \mathbb{D}_i.Location$
                }
                \Else
                {
                    \lIf{$\mathbb{V}'_j.TransportType=MOTORED$}
                    {
                        $Vicinity \leftarrow T_P^{m}$
                    }
                    \lElse
                    {
                        $Vicinity \leftarrow T_P^{nm}$
                    }
                }
                \If(\tcp*[f]{Upgrade volunteer}){$Vicinity > \mathbb{D}_i.Vicinity$}
                {
                    $\mathbb{D}_i.Vicinity \leftarrow Vicinity$\\
                    \lIf{$\mathbb{D}_i \notin \mathbb{M}$}
                    {
                        $\mathbb{M}.AddMatching$($\mathbb{D}_i$, $\mathbb{V}'_j$)
                    }
                    \lElse
                    {
                        $\mathbb{M}.UpdateMatching$($\mathbb{M}.GetMatching(\mathbb{D}_i)$, $\mathbb{V}'_j$)
                    }
                }
            }
            $\mathbb{V} \leftarrow \mathbb{V} - \mathbb{M}.GetMatching(\mathbb{D}_i)$ \tcp*[f]{Remove volunteer from market}
        }
    }
    \ForEach(\tcp*[f]{Update agent preferences}){$\mathbb{A} \in \mathbb{R}' \cup \mathbb{D}'$}
    {
        \If(\tcp*[f]{Receivers}){$\mathbb{A}.RequestType = \mathbb{R}.UserType$}
        {
            $\mathbb{A}.Preference \leftarrow [\{\mathbb{A}.Preference \cap \mathbb{A}.EligibleDonors(\mathbb{D}')\} \cup \{\mathbb{A}.EligibleDonors(\mathbb{D}') - \mathbb{A}.Preference\}]$
        }
        \Else(\tcp*[f]{Donors})
        {
            $\mathbb{A}.Preference \leftarrow [\mathbb{A}.Preference \cap \{(\mathbb{A}.Location \pm \mathbb{A}.Vicinity)  \cap \mathbb{R}'.All.Location \cap (\mathbb{M}.GetMatching(\mathbb{A}).Volunteer.Route \pm T_l \times \mathbb{M}.GetMatching(\mathbb{A}).Volunteer.Route.Distance)\}]$
        }
    }
    \ForEach(\tcp*[f]{Match receivers}){$\mathbb{R}_i \in \{\mathbb{R}_k \vert \mathbb{R}_k.Amount > 0, \mathbb{R}_k \in \mathbb{R}'\}$}
    {
        \While(\tcp*[f]{Meals per requirement}){$\mathbb{R}_i.RemainingAmount > 0$}
        {
            $\mathbb{R}_i.Preference \leftarrow \mathbb{R}_i.Preference \cap \mathbb{D}$\\
            $\mathbb{D}_j \leftarrow \mathbb{R}_i.Preference.PreferredDonor($ $min(\mathbb{R}_i.Preference.PriorityPosition(\mathbb{R}_i)))$\\
            $\mathbb{M}.UpdateMatching(\mathbb{M}.GetMatching(\mathbb{D}_j)$, $\mathbb{R}_i))$ \tcp*[f]{Agents matched}\\
            $\mathbb{D} - \mathbb{D}_j$\tcp*[f]{Remove donor from market}
        }
    }
    \Return $\mathbb{M}$, $\mathbb{D}$, $\mathbb{R}$, $\mathbb{V}$ \tcp*[f]{Return all updated lists}
\end{algorithm}
\DecMargin{1em}

\section{Analysis} \label{Analysis}

In this section we discuss the \textit{correctness}, \textit{strategyproofness}, \textit{Pareto-optimality}, and \textit{(polynomial) running time} properties of the FDRM-CA algorithm. We follow this up with the simulation results and their analysis.

\subsection{Algorithm Analysis}

\begin{lemma}
FDRM-CA works correctly.
\end{lemma}

\begin{proof}[\textbf{Proof:}]
We prove this using the loop invariant technique \citep{cormen2009introduction}. We start our proof with the main FDRM-CA algorithm and then we go into the detailing of the inner sub-processes. FDRM-CA consists of three sub-processes, namely, NUIR, CTFU and the CA-DTB. In this proof, along with the main mechanism of FDRM-CA, the sub-processes are taken into account as well, so that, we can prove that as a whole, FDRM-CA works correctly.

\begin{itemize}
\item \textbf{The Main Routine of FDRM-CA.}\\

\noindent For this algorithm we use the following loop invariant:\\

    \begin{adjustwidth}{1cm}{1cm}
    At the start of the $i^{th}$ iteration of the \textbf{for} loop of lines 5-13, each receiver $\mathbb{R}_j$ processed in the prior $(i - 1)$ iterations have their best available donors allocated, respecting the eligible donors' preference lists.\\
    \end{adjustwidth}

\textbf{Initialization:} Prior to the first iteration of the loop, $i = 1$. There are no receivers processed and no matching exists. This trivially satisfies the invariant.

\textbf{Maintenance:} At any iteration $i$ of the loop, each receiver $\mathbb{R}_j$ processed is classified into either \textit{perishable} or \textit{non-perishable} category. Thereafter, for each donor in the priority list of $\mathbb{R}_j$, the donor's ranking of $\mathbb{R}_j$ is retrieved. The donor $\mathbb{D}_k$ with the highest ranking for $\mathbb{R}_j$ is assigned to $\mathbb{R}_j$. Furthermore, incriminating $i$ for the next iteration of the \textbf{for} loop maintains the invariant.

Of course, the correctness of the above step depends on the correctness of the inner sub-parts which is what has been proven later in this section.

\textbf{Termination:} At termination, when $i = n + 1$, each receiver $\mathbb{R}_j$ processed in the prior $n$ iterations have their best available donors allocated, respecting the eligible donors' preference lists. This proves that the FDRM-CA algorithm works correctly. To reinforce the maintenance step of this proof, following are the additional proofs of correctness of the sub-processes of this algorithm.\\

\item \textbf{The NUIR Sub-Process.}\\

\noindent This can be proved along similar lines as the above.\\

\item \textbf{The CTFU Sub-Process.}\\

\noindent For this algorithm we use the following loop invariant:\\

    \begin{adjustwidth}{1cm}{1cm}
    At the start of the $i^{th}$ iteration of the \textbf{foreach} loop of lines 1-13, the first $(i - 1)$ agent requests from $\mathbb{C}$, processed in the prior $(i - 1)$ iterations, have been classified and appended to one of the five lists, namely, $\mathbb{V}, \mathbb{PFD}, \mathbb{NPFD}, \mathbb{PFR}, \mathbb{NPFR}$.\\
    \end{adjustwidth}

\textbf{Initialization:} Prior to the first iteration of the \textbf{foreach} loop, there are no agent requests processed, and this trivially satisfies the invariant.

\textbf{Maintenance:} In the $i^{th}$ iteration, the $i^{th}$ agent request is taken up for processing from the list $\mathbb{C}$. If this agent request is a volunteer request, it is appended to the list $\mathbb{V}$. If not, it is classified and appended to one of the four food request lists ($\mathbb{PFD}, \mathbb{NPFD}, \mathbb{PFR}, \mathbb{NPFR}$) as per \hyperref[Tab-2]{\textbf{Table 2. Bifurcation of Food Requests}}. Observe that, prior to this iteration, all the five lists, namely, $\mathbb{V}, \mathbb{PFD}, \mathbb{NPFD}, \mathbb{PFR}, \mathbb{NPFR}$, already had a total of $(i - 1)$ agent requests from $\mathbb{C}$ classified and appended to them from the $(i - 1)$ prior iterations. As the \textbf{foreach} loop picks up the next agent request from $\mathbb{C}$ for the next iteration, the loop invariant is reestablished.

\textbf{Termination:} At termination, the list $\mathbb{C}$ is empty and all agent requests have been classified and appended to one of the five lists, namely, $\mathbb{V}, \mathbb{PFD}, \mathbb{NPFD}, \mathbb{PFR}, \mathbb{NPFR}$. Thus, the invariant is maintained. This proves that the CTFU process works correctly.\\

\item \textbf{The CA-DTB Sub-Process.}\\

\noindent This can be proved along similar lines as the above.
\end{itemize}

By having proven the correctness of all its sub-parts, this concludes our proof of correctness of the FDRM-CA mechanism.
\end{proof}

\begin{lemma}
FDRM-CA is strategyproof for donors and receivers.
\end{lemma}

\begin{proof}[\textbf{Proof:}]
While there can be many places in FDRM-CA for donors to misreport, however, they can never gain from the same.

Firstly, say a donor $\mathbb{D}_i$ misreports the donor's preference list as $\mathbb{R}_k \succ \mathbb{R}_j$. Since the CA-DTB method awards the best available option from this preference list, therefore, misreporting only gets the donor $\mathbb{D}_i$ an option $\mathbb{R}_k$ which might not be the best available as per the donor's truthful preference list. The only way a donor $\mathbb{D}_i$ can gain is if $\mathbb{R}_k$ becomes unavailable in the \textit{extracted and updated} preference list of $\mathbb{D}_i$. However, that is a run-time phenomenon not affected by the preference list of the donor, and hence, not affected by manipulation. The above line of logic can be extended to other input details of the donors like location and type of food. Also, since an ICT based platform is used for tracking donations, it eliminates misreporting of factual data from the donors' end. Note that misreporting the amount of food donated by $\mathbb{D}_i$ does not affect the the order in which receivers are chosen for matching, it only changes the count of receivers assigned to the donor $\mathbb{D}_i$. Secondly, donors cannot influence the choice of their neighbourhood as it depends on the run-time availability of volunteers ($\mathbb{V}$) and the details of the transport used. Thirdly, ties are broken primarily in favour of the earliest food donation or requirement timings which can be misreported. Consider a donor $\mathbb{D}_i$ who misreported food pickup time to a time ($t_a$) before the food is available for pickup ($t_b$; $t_a<t_b$). Since, food pickup is a volunteer mediated activity and volunteers only have a window of time available, therefore, the availability of a volunteer $\mathbb{V}_j$ at a misreported time ($t_a$) does not guarantee the donor the availability of any volunteer at the donor's truthful pickup time ($t_b$). Worst case scenario, no volunteers may be available at the truthful pickup time ($t_b$). This concludes the proof.

It can be proved, along similar lines, that FDRM-CA is strategyproof for receivers as well.
\end{proof}

\begin{lemma}
FDRM-CA is Pareto-optimal for donors.
\end{lemma}

\begin{proof}[\textbf{Proof:}]
This we prove by contradiction and mathematical induction. Let us assume that FDRM-CA is not Pareto-optimal. Also, let there be some other algorithm, say OTH, that assigns Pareto-optimally.

The basis condition for induction is trivial for zero iterations. For the induction step, say until an iteration $i$, both algorithms generate a similar matching $\mathbb{M}_i$. In the $(i + 1)^{th}$ iteration, the CA-DTB sub-process of the FDRM-CA mechanism assigns the best available option for donors as per their respective preference lists and updates the matching to $\mathbb{M}_{i+1}$. If OTH assigns anything other than the above, then it does not generate a Pareto-optimal matching ($\mathbb{M}_{i+1}$) at this iteration. Therefore, by the Principle of Mathematical Induction, the final matching generated by OTH will also not be optimal. This is a direct contradiction to our initial assumption. Thus, we can say that our initial assumption was not correct and that FDRM-CA is Pareto-optimal for donors. This concludes the proof.
\end{proof}

\begin{lemma}
FDRM-CA produces results in real (polynomial) time.
\end{lemma}

\begin{proof}[\textbf{Proof:}]
We only analyze the time complexity of a single iteration of the FDRM-CA mechanism since after each iteration agents get their assignments. Since food is only recoverable inside a time window, we will focus our analysis from the donors' point of view. For this, we analyze the time complexity of each sub-process of the mechanism. Note that receivers with large requirements may get their complete matching (that is equal to their total requirement) after several iterations of the mechanism.
    \begin{itemize}
        \item The NUIR sub-process is parallel to the main process and thus does not have any impact on the time complexity as it does not need to finish for the main process to execute,
        \item The CTFU sub-process takes O$(d+r+v)<$ O$(r^2)$,
        \item The CA-DTB sub-process can be viewed as having two sub-parts:
        \begin{itemize}
            \item The sub-part to get the neighbourhood of each donor and volunteer assignment to the donor takes O$(dv)$ $<$ O$(r^3)$,
            \item The update of preference and match generation sub-part has a time complexity of O$(d^2r)$ $<$ O$(r^3)$,
        \end{itemize}
        since $r>d$ for most practical situations, especially in developing countries like India, at the time of writing this paper,
    \end{itemize}
where $d$, $r$, and $v$ are the numbers of donors, receivers, and volunteers respectively, per iteration of the FDRM-CA mechanism. Therefore, the effective time complexity of each iteration of the FDRM-CA mechanism is capped at O$(r^3)$ which is polynomial (real) time. This concludes the proof.
\end{proof}

\begin{lemma}
Using an off routing threshold $T_l\%$, the maximum total off-routing percentage $(\Gamma)$ for any volunteer for each meal transportation will always be less than equal to four times this threshold percent of the route distance of the corresponding volunteer, i.e., $\Gamma \leq 4 \times T_l$.
\end{lemma}

\begin{proof}[\textbf{Proof:}]
Analyzing a worst case scenario, let the donor $\mathbb{D}_i$'s location be opposite to the volunteer $\mathbb{V}_j$'s route, and at a maximum possible distance of $T_l\%$ of the volunteer's travel distance. Therefore, the donation pickup will contribute to an extra travel distance of $2 \times T_l\%$ for the volunteer. At the drop-off end, let us again assume a worst case scenario by allowing the receiver $\mathbb{R}_k$ to be at a 180\textdegree ~off-route location yielding another $2 \times T_l\%$ travel overhead for the volunteer $\mathbb{V}_j$. Thus, even in the worst case scenario, the volunteer $\mathbb{V}_j$ has to do a maximum total off-routing percentage given by the inequality $\Gamma \leq 4 \times T_l$. This concludes the proof.
\end{proof}

\subsection{Algorithm Simulation}

In this section we produce the simulation results of the algorithm and analyze them.

\subsubsection{Data}

We generate random, logically coherent data for $5000$ agent requests for the simulation. We have chosen to model our operational city to be $50 \times 50$ kilometers across, our working hours to be from 06:00 hours to 23:59 hours, any volunteer's maximum payload capacity to be 100 kilograms. We have taken our threshold values as: $T_o=0.25$ hour, $T_d=2$ hours, $T_r=3$ hours, $T_w=0.25$ hour, $T_l=5\%$, $T_a=20\%$, $T_m=1$ kilogram, $T_P^m=20$ kilometers, $T_P^{nm}=5$ kilometers, and $T_{NP}=100$ kilometers.

\subsubsection{Results}

The simulation results have been presented in \hyperref[Fig-8]{\textbf{Figure 8. Simulation Results}}. In all the graphs, volunteer numbers have been expressed in terms of multiples of donation numbers ($\times Donors$). From \hyperref[Fig-8a]{\textbf{Figure 8(a). Agent Allocation vs Volunteer Availability}}, it is evident that the allocation percentage of agents initially increases rapidly with increase of volunteer availability and then plateaus as the availability grows further. \hyperref[Fig-8b]{\textbf{Figure 8(b). Start vs End Sorting for Receivers}} supports our initial claim that volunteers going for the receiver having the earliest requirement end time first will be able to address more requests. Similarly, \hyperref[Fig-8c]{\textbf{Figure 8(c). Original vs Eligible Preferences for Agents}} depicts the importance of updating the submitted agent preferences to reflect run-time temporal and spatial availability of agents. At last, \hyperref[Fig-8d]{\textbf{Figure 8(d). Results of Manipulation of Agent Preferences}} establishes that manipulation of preferences will not help agents to gain a better allocation other than when a small percentage of them do not have their false higher preferred donors unavailable while calculating their eligibility lists. In the first graph of \hyperref[Fig-8a]{\textbf{Figure 8(a). Agent Allocation vs Volunteer Availability}}, actual spacing between volunteer values have been maintained for volunteer availability to reveal the true nature of the allocation curve. \hyperref[Fig-8b]{\textbf{Figure 8(b). Start vs End Sorting for Receivers}} uses agent preferences updated with their eligibility list, and \hyperref[Fig-8c]{\textbf{Figure 8(c). Original vs Eligible Preferences for Agents}} uses end sorting for receivers. \hyperref[Fig-8d]{\textbf{Figure 8(d). Results of Manipulation of Agent Preferences}} has both end sorting and updated preferences running in the background.

\begin{figure}
    \centering
    \subfloat[Allocation vs Volunteer Availability]
        {\includegraphics[width=0.48\textwidth, height=4.2cm]{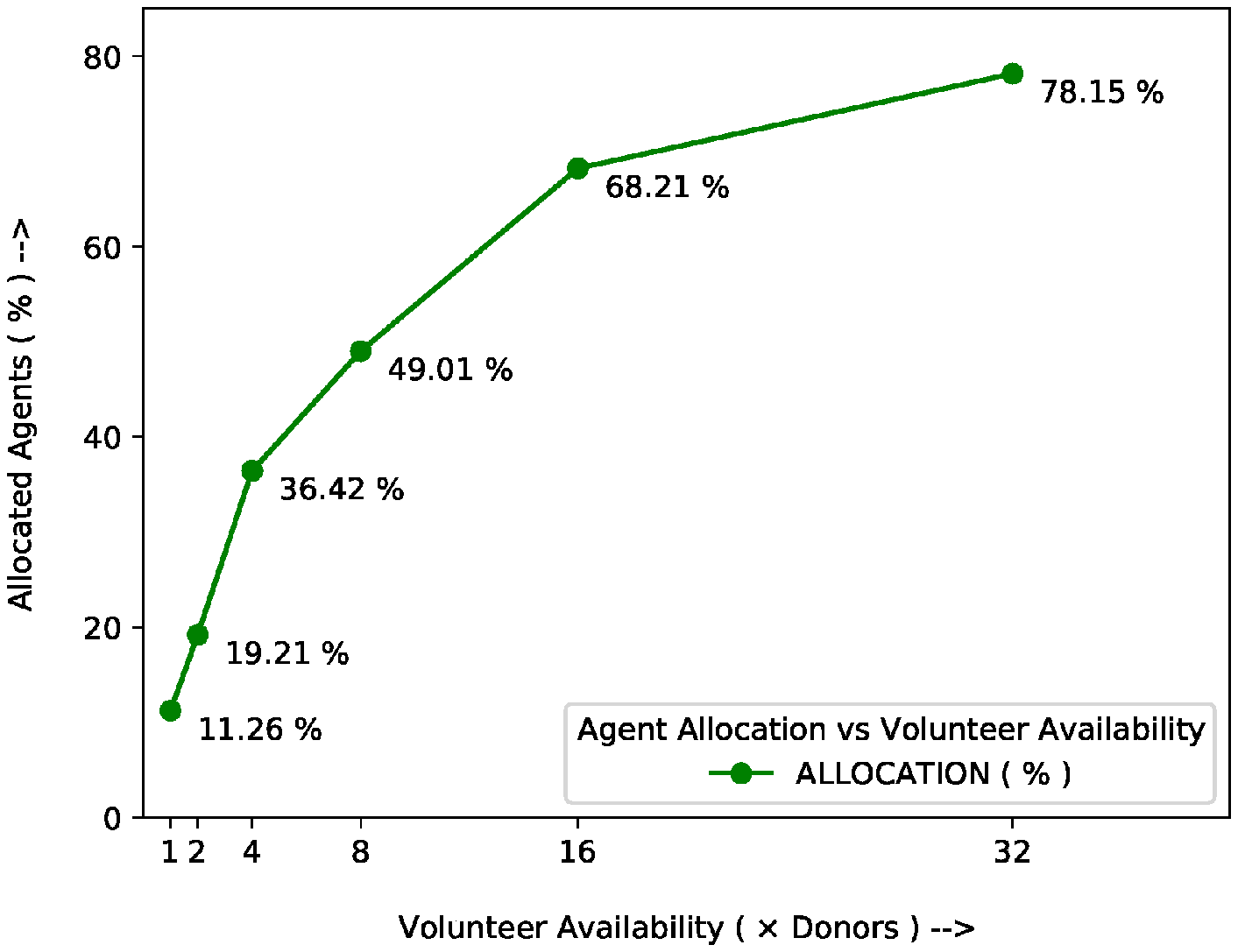}
        \label{Fig-8a}}
    \subfloat[Start vs End Sorting for Receivers]
        {\includegraphics[width=0.48\textwidth, height=4.2cm]{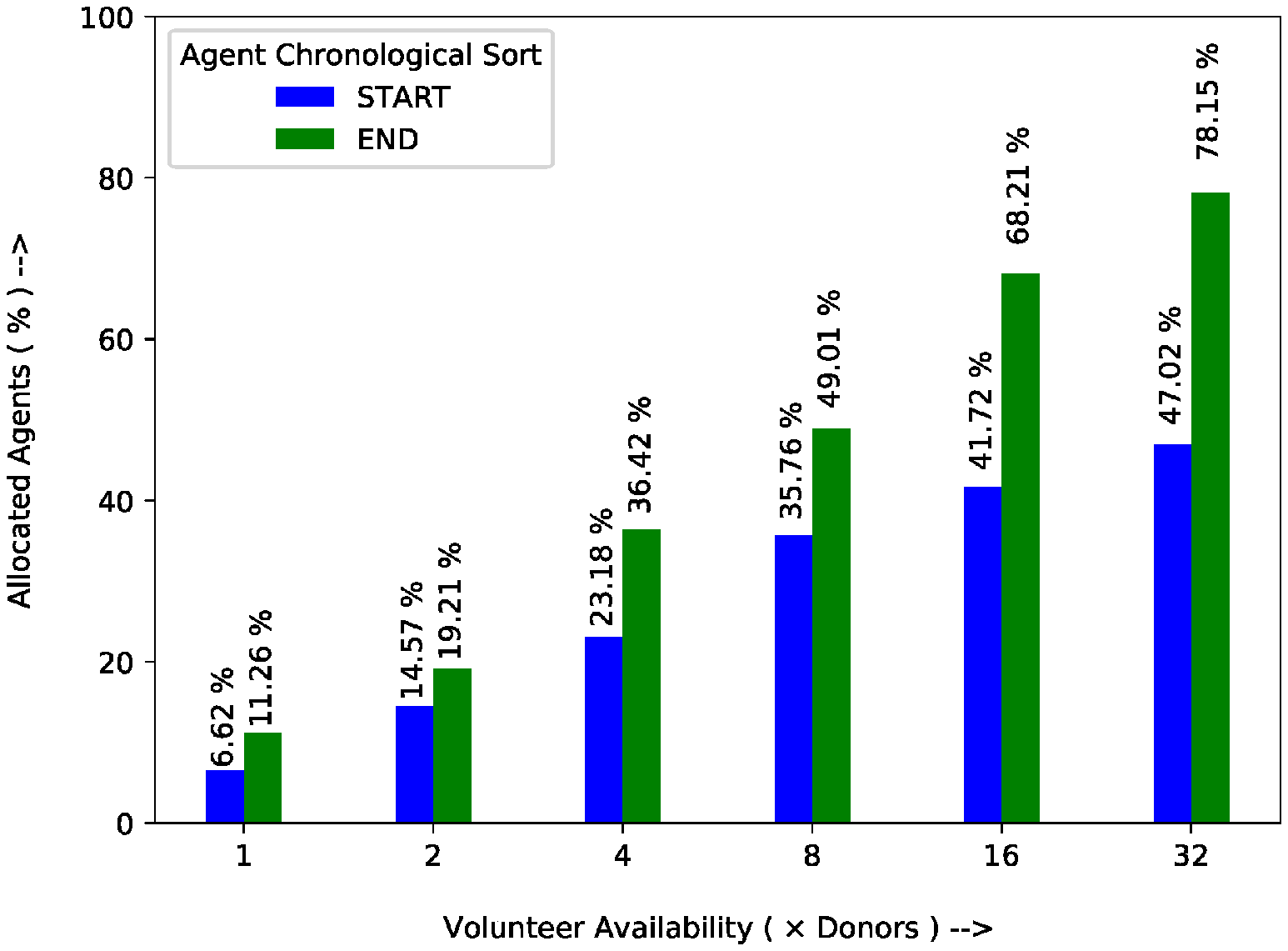}
        \label{Fig-8b}}
        \newline
    \subfloat[Original vs Eligible Preferences for Agents]
        {\includegraphics[width=0.48\textwidth, height=4.2cm]{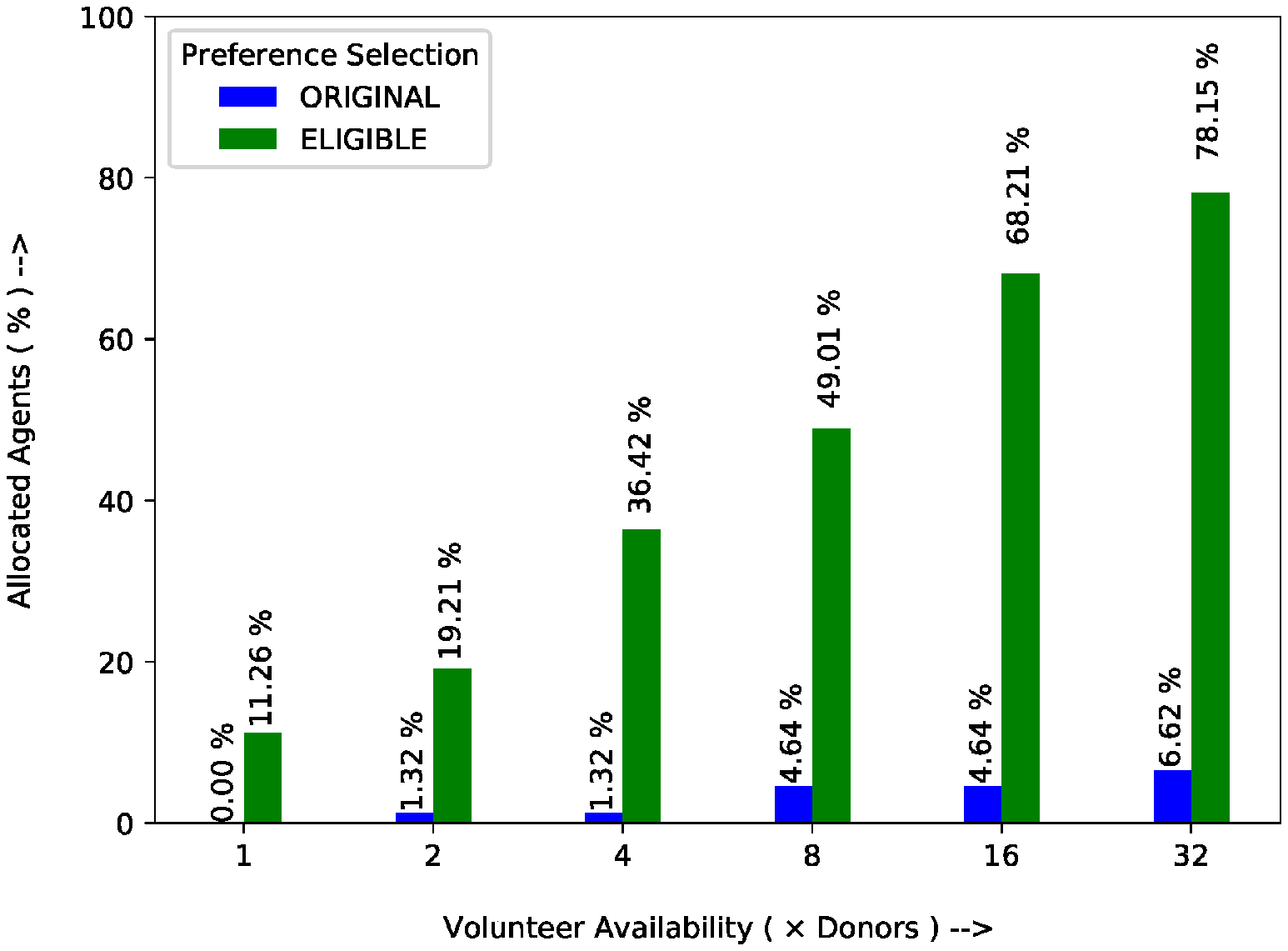}
        \label{Fig-8c}}
    \subfloat[Agent Preference Manipulation Results]
        {\includegraphics[width=0.48\textwidth, height=4.2cm]{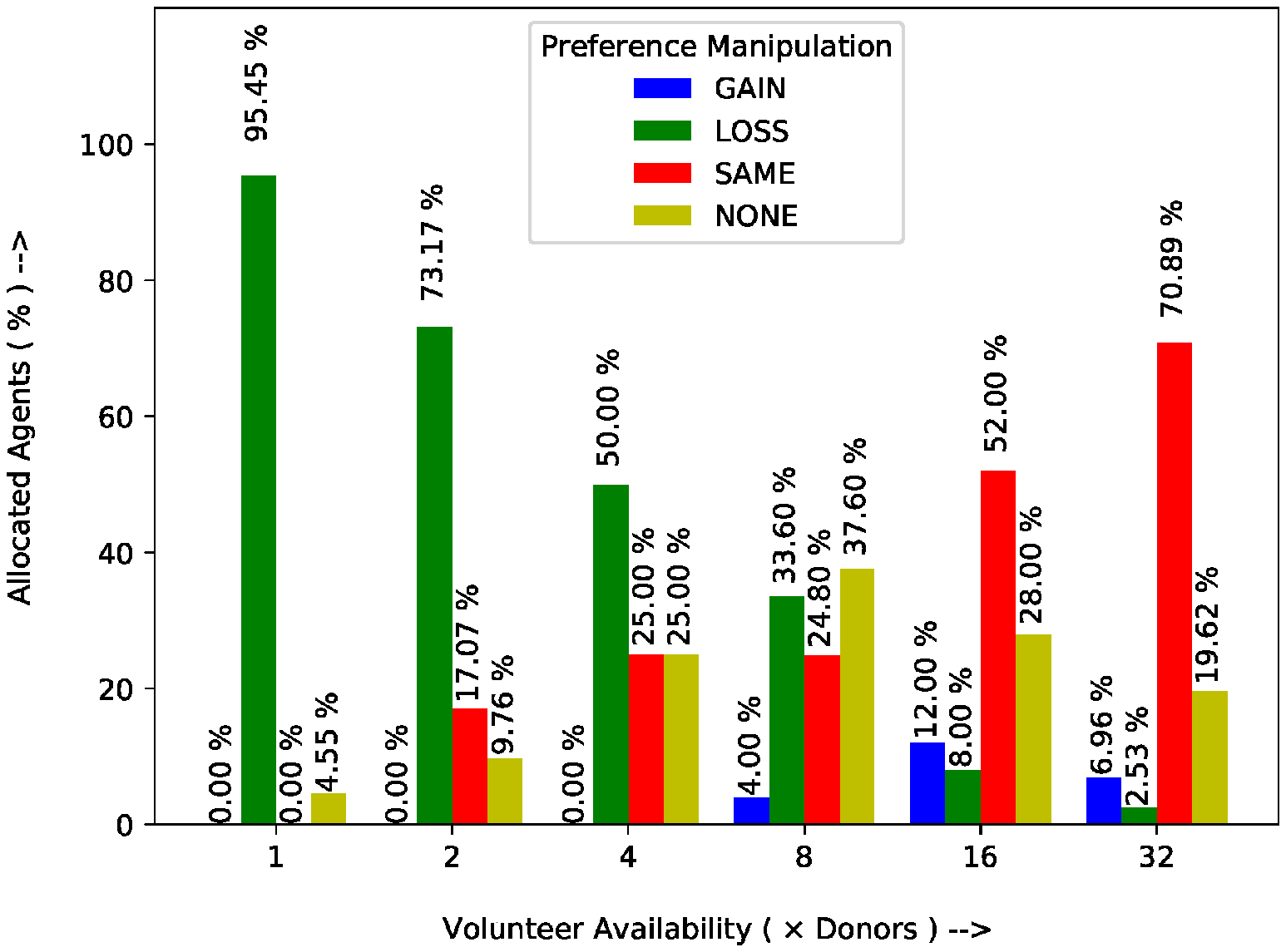}
        \label{Fig-8d}}
    \caption{Simulation Results}
    \label{Fig-8}
\end{figure}

\section{Conclusion and Future Work} \label{Sec-7}

\subsection{Conclusion}

Food redistribution is visibly not a long term solution for food wastage, and it has to be tackled at the roots by putting a check on the inclination of the society towards a permanent over-supply of food, and a fear of ever running-out of food. However, until this society structure is fixed, ICT will continue to bridge the gap between the over-supplied and the under-provided. Although, 100\% of the food waste cannot be targeted by this approach, the 60\% avoidable waste will be well diverted from landfills to the under-provided people. The FDRM-CA mechanism, delivered over handheld device apps, provides a great platform for doing this specifically. It prevents donors from misreporting, addresses donor and receiver preferences properly, prioritizes food donation/requirement events chronologically, matches donors optimally with receivers and volunteers, and above all, does all these to provide agents with their respective assignments in real time, making it an attractive choice for the task. As receivers broaden the spectrum of their requirements from the more common \textit{freshly cooked} or \textit{packaged solid} food types towards the statistically rarer \textit{fresh produce} or \textit{frozen uncooked} food types, and as donors get diversified from households to farmers and businesspersons gradually, this redistribution approach will be able to target the different stages of food production wherein food gets wasted. With the use of the proposed mechanism, the preferences of all the parties can be addressed better, leading to a higher satisfaction of all the agents involved. This is guaranteed to attract more participation into the food redistribution network, eventually lowering the food wastage and hunger challenges of the world for a sustainable and scalable food future.

\subsection{Future Work}

In the future, if some volunteers find incentives as their primary motivation to participate in the system, the same can be integrated with the model for a hybrid functionality so as to attract a broader volunteer participation leading to an overall better surplus movement.

\end{document}